\documentclass[11pt,a4paper,oneside]{article}

\usepackage[T1]{fontenc}
\usepackage[utf8]{inputenc}
\usepackage{microtype}
\usepackage[margin=1in]{geometry}

\usepackage{fancyhdr}
\newcommand{\runningtitle}{MCSI: Masked Commutative Supersingular Isogeny Key Exchange}

\usepackage{amsmath,amssymb,amsthm,mathtools}
\usepackage{needspace}
\usepackage{newtxtext,newtxmath}
\usepackage{bm}

\usepackage{algorithm}
\usepackage{algpseudocode}
\usepackage{booktabs}
\usepackage{graphicx}
\usepackage{tikz}
\usetikzlibrary{arrows.meta,positioning,calc,decorations.pathmorphing}
\usepackage{pgfplots}
\usepgfplotslibrary{groupplots}
\pgfplotsset{compat=1.18}

\usepackage[hidelinks]{hyperref}
\usepackage{cleveref}

\theoremstyle{plain}
\newtheorem{theorem}{Theorem}[section]
\newtheorem{lemma}[theorem]{Lemma}
\newtheorem{proposition}[theorem]{Proposition}
\newtheorem{corollary}[theorem]{Corollary}

\theoremstyle{definition}
\newtheorem{definition}[theorem]{Definition}
\newtheorem{problem}[theorem]{Problem}
\newtheorem{assumption}[theorem]{Assumption}

\theoremstyle{remark}
\newtheorem{remark}[theorem]{Remark}

\newcommand{\ZZ}{\mathbb{Z}}
\newcommand{\QQ}{\mathbb{Q}}

\newcommand{\Fp}{\mathbb{F}_p}
\newcommand{\Fpsq}{\mathbb{F}_{p^2}}
\newcommand{\Fpbar}{\overline{\mathbb{F}}_p}

\newcommand{\Eprime}{E'}
\DeclareMathOperator{\End}{End}

\DeclareMathOperator{\ord}{ord}
\DeclareMathOperator{\disc}{disc}
\DeclareMathOperator{\Norm}{N}
\newcommand{\Ord}{\mathcal{O}}              
\newcommand{\clgp}{\mathrm{cl}(\Ord)}       
\newcommand{\ida}{\mathfrak{a}}
\newcommand{\idb}{\mathfrak{b}}

\newcommand{\idl}{\mathfrak{l}}
\newcommand{\Ellset}{\mathcal{E}\ell\ell_p(\Ord,\pi)}
\newcommand{\isog}{\varphi}
\newcommand{\dualisog}{\hat{\varphi}}
\newcommand{\Ptinf}{\mathcal{O}_E}          
\newcommand{\frob}{\pi}

\newcommand{\Etors}[1]{E[#1]}

\newcommand{\act}{\star}
\newcommand{\Gact}{G}
\newcommand{\Xset}{X}
\newcommand{\Ezero}{E_0}
\newcommand{\Mont}[1]{E_{#1}}

\newcommand{\KeyGen}{\mathrm{KeyGen}}

\newcommand{\Validate}{\mathrm{Validate}}
\newcommand{\MCSI}{\ensuremath{\mathrm{MCSI}}}
\newcommand{\Mask}{\mathrm{Mask}}
\newcommand{\Unmask}{\mathrm{Unmask}}
\newcommand{\AEnc}{\mathrm{Enc}}
\newcommand{\ADec}{\mathrm{Dec}}
\newcommand{\Mac}{\mathrm{MAC}}
\newcommand{\encode}{\mathrm{enc}}
\newcommand{\decode}{\mathrm{dec}}
\newcommand{\Itos}{\mathrm{I2OSP}}

\newcommand{\secpar}{\lambda}
\newcommand{\negl}{\mathrm{negl}}
\newcommand{\Adv}{\mathcal{A}}
\newcommand{\Bdv}{\mathcal{B}}
\newcommand{\Hash}{\mathrm{H}}

\newcommand{\cat}{\mathbin{\|}}
\newcommand{\getsr}{\stackrel{\$}{\leftarrow}}
\newcommand{\advantage}[2]{\mathbf{Adv}^{\mathrm{#1}}_{#2}}
\newcommand{\bits}{\{0,1\}}
\newcommand{\ctxt}{\mathit{ctx}}
\newcommand{\tr}{\mathit{tr}}

\title{MCSI: A Masked Commutative Supersingular Isogeny\\ Key Exchange with Blinded Ephemeral Keys}

\author{%
\begin{tabular}{@{}c@{\hspace{0.8cm}}c@{\hspace{0.8cm}}c@{}}
Furkan Cifci\textsuperscript{$\star$} & Osman Emre Donder & Reyyan Cifci \\[2pt]
\small M.Emin Sarac High School & \small Computer Engineering & \small Computer Engineering \\
\small Istanbul, Turkiye & \small Bilkent University & \small King Fahd University of \\
 & \small Ankara, Turkiye & \small Petroleum and Minerals \\
\small \textsuperscript{$\star$}\,Corresponding author & & \small Saudi Arabia \\[3pt]
\small \texttt{thefurkancifci@gmail.com} & \small \texttt{a.o.bafrali@gmail.com} & \small \texttt{spherosic@gmail.com}
\end{tabular}%
}

\date{\today}

\begin{document}
\maketitle

\begin{abstract}
We introduce \MCSI{}, a two message key exchange we design over the CSIDH class group action, in
which each party sends its ephemeral public element under an authenticated encryption keyed by the
value the two static keys determine. The design gives implicit mutual authentication, hides the
ephemeral element from an eavesdropper, and lets a recipient discard an unauthenticated message
after one tag check rather than after an evaluation of the group action, which is four orders of
magnitude more expensive.

On the analytic side, we prove that our protocol is correct with zero error, and prove three
statements in the random oracle model, all reducing to the strong parallelisation problem:
indistinguishability of the session key against a passive adversary, confidentiality of the
blinded ephemeral element, and integrity of the blinded transport. None uses the decisional group
action assumption, which is false for class group actions of non-prime discriminant. We also show
that a blinding key cannot come from the session secret it is meant to establish.

To instantiate the design we select parameters and show that a prime chosen for elliptic curve
discrete logarithms is unusable: for $p = 2^{521}-1$, the NIST P-521 prime, the action admits no
efficiently evaluable generator.

On the practical side, we build and test the design. We implement the protocol twice, in C and
independently in Python, cross check the two, and measure what a session costs in field
operations, time and memory. We also audit our code for secret dependent control flow: the field
arithmetic and the symmetric layer show none, while the group action leaks the key by
construction, and two hundred timings separate two keys whose one-norms differ by five out of 370.

Finally, we state what we do not prove, among them security under ephemeral key reveal, forward
secrecy of the blinding, and constant time execution.

\medskip
\noindent\textbf{Keywords:} post-quantum cryptography; isogeny based cryptography; CSIDH;
cryptographic group actions; authenticated key exchange.
\end{abstract}
\begingroup
\setcounter{tocdepth}{1}
\setlength{\parskip}{0pt}
\needspace{15\baselineskip}
\tableofcontents
\endgroup
\section{Introduction}
\label{sec:intro}

A cryptographically relevant quantum computer would break the public key cryptography that
secures most network traffic today, and although no such machine exists, the migration has to
begin long before one does, because data recorded now can be decrypted later~\cite{Mosca2018}.
The response has been a decade of work on post-quantum
cryptography~\cite{BernsteinLange2017} and, since 2024, a standardised lattice based key
encapsulation mechanism in FIPS 203~\cite{FIPS203}. Standardisation on a single mathematical
family is uncomfortable, and the interest in alternatives with different underlying problems
is not merely academic.

Isogeny based cryptography is one such alternative, and it has the smallest public keys of any
post-quantum family by a wide margin. It also had, in 2022, the most dramatic failure. The
following two paragraphs matter for how this paper positions itself, so we state them
carefully.

\subsection{The isogeny landscape after the SIDH break}
\label{sec:landscape}

There are two branches. The first, initiated by Jao and De
Feo~\cite{JaoDeFeo2011,DeFeoJaoPlut2014}, works with supersingular curves over $\Fpsq$, where
the endomorphism ring is an order in a quaternion algebra and is therefore noncommutative. To
recover enough structure for a Diffie--Hellman analogue, SIDH publishes the images of a
torsion basis under the secret isogeny. That auxiliary data proved fatal: Castryck and
Decru~\cite{CastryckDecru2023} recovered keys from it in 2022, and Maino et
al.~\cite{Maino2023} and Robert~\cite{Robert2023} generalised the attack, the last of these to
a polynomial time algorithm with no remaining restrictions. SIKE, the SIDH based candidate in
the NIST process, was withdrawn.

The second branch, going back to Couveignes~\cite{Couveignes2006} and Rostovtsev and
Stolbunov~\cite{RostovtsevStolbunov2006} and made practical for supersingular curves by
Castryck, Lange, Martindale, Panny and Renes~\cite{CSIDH2018}, works with curves defined over
the prime field $\Fp$. There the $\Fp$-rational endomorphism ring is an order in an imaginary
quadratic field and is commutative, and the ideal class group of that order acts freely and
transitively on the relevant set of curves. A Diffie--Hellman analogue follows from the
commutativity of the class group. This branch, CSIDH, transmits nothing but a single field
element per public key and publishes no torsion data at all, and it was untouched by the 2022
attacks. Its weaknesses are elsewhere: a single evaluation of the action takes tens of
milliseconds, and the concrete quantum security of the smaller parameter sets is disputed, with
published estimates spanning a wide range~\cite{Peikert2020,BonnetainSchrottenloher2020,SQALE2022}.

We stress one consequence of this history, because a design in this area is easy to motivate
incorrectly. The SIDH break was caused by publishing torsion images. CSIDH publishes none, so
there is no torsion disclosure in the commutative branch to defend against, and any scheme in
that branch that presents itself as closing that particular gap is describing a gap it never
had. Work on hiding torsion images is real and useful, and M-SIDH and
MD-SIDH~\cite{MSIDH2023} are the developed version of it, but it belongs to the first branch.
This paper belongs to the second, and the thing it hides is not torsion data.

\subsection{What this paper does}
\label{sec:contributions}

Consider a key exchange built directly on the CSIDH action. Each party sends a curve, applies
its secret to the curve it receives, and hashes the result. The expensive step, by a wide margin, is the evaluation of the action, and in such a protocol a party performs
that evaluation on data supplied by whoever sent the packet. It also validates that data first,
which is cheaper but not free. An unauthenticated responder therefore does tens of
milliseconds of work on behalf of anybody who can reach it, and it does so on an element of
the curve set that the sender chose.

\MCSI{}, for masked commutative supersingular isogeny, changes this. Both parties hold static
key pairs, and the static-static shared value $Z_{\mathrm{ss}}$, which by commutativity both
can compute and which can be cached once per peer, keys an authenticated encryption of the
ephemeral element carried in each message. A recipient that cannot verify the tag stops after
one hash and one tag check, having evaluated no group action; a recipient that can verify the
tag knows that the sender holds the peer's static secret. The same layer hides the ephemeral
element from an observer, and the resulting protocol is implicitly mutually authenticated. The
word masked refers to this blinding of the ephemeral public element, and to nothing else.

Our contributions are the following.

\begin{itemize}
  \item A specification of \MCSI{} over an abstract effective group action
        (\Cref{sec:protocol}), together with the design argument that fixes it. In particular
        we show that a blinding key cannot be derived from the session secret it is supposed
        to help establish, enumerate the alternatives, and justify the one we take
        (\Cref{sec:rationale}).
  \item A correctness theorem with a proof that identifies exactly the algebraic facts used,
        namely commutativity of the class group, the action axioms, and uniqueness of the
        Montgomery representative (\Cref{thm:correctness}). We also record which of these
        fails for the iterated isogeny walk that a first attempt at such a protocol tends to
        produce.
  \item Three security theorems, all in the random oracle model and all reducing to the
        strong parallelisation assumption: indistinguishability of the session key against a
        passive adversary (\Cref{thm:sk-ind}), confidentiality of the blinded ephemeral
        element (\Cref{thm:blinding}), and integrity of the blinded transport
        (\Cref{thm:integrity}). None of them uses the decisional assumption, which for this
        instantiation is false~\cite{CSV2020}.
  \item A short lemma showing that a fixed public byte substitution applied to ciphertexts,
        a layer sometimes added to constructions of this shape, leaves every security notion
        exactly where it was (\Cref{lem:sbox}).
  \item A parameter analysis (\Cref{sec:parameters}), including the observation that primes
        chosen for elliptic curve discrete logarithms are unusable here. For the concrete case
        $p = 2^{521}-1$, the prime of NIST P-521, the class group action cannot be evaluated
        at all, because $p+1$ is a power of two and the only available generator would need an
        exponent of size $2^{260}$ (\Cref{rem:2521}).
  \item An explicit list of the properties we do not prove (\Cref{sec:limits}).
  \item A reference implementation, in portable C and in Python, checked against each other
        by known answer vectors and against published test vectors for the hash primitives,
        with measured costs and a measured account of how far its running time follows the
        private key (\Cref{sec:impl}).
\end{itemize}

We claim no new hardness assumption, no speed advantage over any existing scheme, and no
novelty for the idea of building an authenticated key exchange on a group action, which is
already done in~\cite{deKock2021,Kawashima2021} with stronger security models than ours. The
measurements we report in \Cref{sec:impl} are of our own implementation, which is neither
constant time nor optimised, and we say what that means for how they should be read. Every
size in this paper is derived from the parameter choices and is labelled as such.

\subsection{Organisation}
\label{sec:organisation}

\Cref{sec:prelim} fixes notation and states the mathematical background, including the
definition of supersingularity, V\'elu's formulae with their domain of validity, the class
group action, and the hard problems we assume. \Cref{sec:protocol} gives the design rationale
and the protocol. \Cref{sec:correctness} proves correctness and \Cref{sec:security} proves the
security statements and surveys the known attacks. \Cref{sec:parameters} selects parameters,
\Cref{sec:impl} describes the reference implementation and reports sizes, costs and the
measured timing leak, \Cref{sec:related} places
the work in the literature, and \Cref{sec:conclusion} lists what remains open.
\section{Preliminaries}
\label{sec:prelim}

\subsection{Notation}
\label{sec:notation}

We write $\secpar$ for the security parameter and $\bits^{n}$ for the set of bit strings of
length $n$. Concatenation of strings is written $x \cat y$. Sampling $x$ uniformly at random
from a finite set $S$ is written $x \getsr S$. A function is negligible in $\secpar$ if it
decays faster than the inverse of every polynomial, and we write $\negl(\secpar)$ for such a
function.

Throughout, $p$ denotes a prime with $p > 3$, $\Fp$ the field with $p$ elements, $\Fpbar$ an
algebraic closure of $\Fp$, and $\frob$ the $p$-power Frobenius endomorphism. Elliptic curves
are written $E$, their point at infinity $\Ptinf$, and their $n$-torsion subgroup
$\Etors{n} = \{P \in E(\Fpbar) : nP = \Ptinf\}$. We reserve $\Ord$ for an order in an
imaginary quadratic field and $\clgp$ for its ideal class group, so that $\Ptinf$ and $\Ord$
denote different objects and are never interchanged. Ideals of $\Ord$ are written in Fraktur,
$\ida, \idb, \idl$, and their classes in square brackets, $[\ida]$.

$\Hash$ denotes a hash function modelled as a random oracle in the security analysis, and
$\Itos(i,k)$ denotes the big-endian encoding of the non-negative integer $i$ into exactly $k$
bytes.

\subsection{Elliptic curves over finite fields}
\label{sec:curves}

An elliptic curve over $\Fp$ is a smooth projective curve of genus one together with a
distinguished rational point $\Ptinf$. Since $p > 3$, every such curve admits a short
Weierstrass model
\begin{equation}
  E : y^2 = x^3 + a_4 x + a_6, \qquad a_4, a_6 \in \Fp ,
  \label{eq:short-weierstrass}
\end{equation}
and the model defines a smooth curve exactly when its discriminant is nonzero, that is when
$4a_4^3 + 27 a_6^2 \neq 0$ in $\Fp$. The chord and tangent construction turns
$E(\Fp) \cup \{\Ptinf\}$ into a finite abelian group with $\Ptinf$ as the identity. We assume
familiarity with the group law and do not reproduce it; see Silverman~\cite{Silverman2009} or
Washington~\cite{Washington2008}.

Two facts are used repeatedly. First, by Hasse's theorem the trace of Frobenius
$t = p + 1 - \#E(\Fp)$ satisfies $|t| \leq 2\sqrt{p}$. Second, for every $n$ coprime to $p$
one has $\Etors{n} \cong (\ZZ/n\ZZ)^2$ as a group, so that $\Etors{n}$ has a basis consisting
of two points of order $n$. Bases of this kind are the objects that SIDH transmits and whose
disclosure was fatal to that scheme; they play no role in the protocol of this paper, a point
we return to in \Cref{sec:torsion-attacks}.

We will also use the Montgomery model
\begin{equation}
  \Mont{A} : y^2 = x^3 + A x^2 + x, \qquad A \in \Fp, \ A^2 \neq 4 ,
  \label{eq:montgomery}
\end{equation}
because it gives a canonical one field element representative for the curves that occur in
our protocol (\Cref{prop:montgomery-normal-form}).

\subsection{Supersingular curves}
\label{sec:supersingular}

The following is the definition of supersingularity. It is a statement about $p$-torsion, and
it is unrelated to whether the defining Weierstrass equation is smooth.

\begin{definition}[Supersingular curve]
\label{def:supersingular}
Let $E$ be an elliptic curve over a field of characteristic $p > 0$. Then $E$ is
\emph{supersingular} if
\[
  \Etors{p}(\Fpbar) = \{ \Ptinf \} ,
\]
that is, if $E$ has no nontrivial $p$-torsion over the algebraic closure. Otherwise
$\Etors{p}(\Fpbar) \cong \ZZ/p\ZZ$ and $E$ is called \emph{ordinary}.
\end{definition}

\begin{proposition}[Equivalent characterisations]
\label{prop:supersingular-equiv}
Let $E$ be an elliptic curve defined over $\Fp$ with $p > 3$, and let $t$ denote the trace of
Frobenius. The following are equivalent.
\begin{enumerate}
  \item $E$ is supersingular in the sense of \Cref{def:supersingular}.
  \item $p \mid t$, which over the prime field $\Fp$ forces $t = 0$ and hence
        $\#E(\Fp) = p + 1$.
  \item $\End_{\Fpbar}(E)$ is an order in a quaternion algebra ramified exactly at $p$ and
        at infinity.
  \item $j(E) \in \Fpsq$.
\end{enumerate}
\end{proposition}

\begin{proof}
This is classical; see~\cite[V.3.1]{Silverman2009} for the equivalence of (1), (3) and (4).
For (2), the $p$-torsion is trivial exactly when the dual of Frobenius is inseparable, which
happens exactly when $p \mid t$; over the prime field the Hasse bound $|t| \leq 2\sqrt{p}$ then
leaves $t = 0$ as the only multiple of $p$ in range once $p > 4$, whence
$\#E(\Fp) = p+1$.
\end{proof}

\begin{remark}[A definition that is sometimes confused with supersingularity]
\label{rem:not-supersingular}
The condition $4a_4^3 + 27a_6^2 \neq 0$ appearing after
\Cref{eq:short-weierstrass} is the condition that the Weierstrass equation defines a
\emph{smooth} curve. Every elliptic curve satisfies it by definition, ordinary curves
included, so it cannot distinguish supersingular curves from any others. A concrete
illustration: over $\Fp$ with $p = 101$, the curve $y^2 = x^3 + 2x + 3$ has discriminant
$4 \cdot 2^3 + 27 \cdot 3^2 = 275 \equiv 73 \not\equiv 0 \pmod{101}$ and is therefore smooth,
but $\#E(\Fp) = 96 \neq 102 = p+1$, so by
\Cref{prop:supersingular-equiv}\,(2) it is ordinary. We state this explicitly because the
distinction is the foundation of everything that follows: the commutativity that our protocol
relies on comes from the $\Fp$-rational endomorphism ring of a \emph{supersingular} curve
over the prime field, and is simply absent for an ordinary curve.
\end{remark}

Supersingular curves over $\Fp$ exist for every $p$, and $E_0 : y^2 = x^3 + x$ is
supersingular whenever $p \equiv 3 \pmod 4$~\cite{CSIDH2018}. This is the starting
curve we use.

\subsection{Isogenies and V\'elu's formulae}
\label{sec:isogenies}

An \emph{isogeny} $\isog : E \to \Eprime$ between elliptic curves over a field $k$ is a
nonconstant morphism of varieties that maps $\Ptinf$ to the point at infinity of $\Eprime$.
Every isogeny is a group homomorphism on points. Its degree is its degree as a morphism, and
for a separable isogeny this equals $\# \ker \isog$. Every isogeny $\isog$ of degree $d$ has a
dual $\dualisog$ with $\dualisog \circ \isog = [d]$.

Separable isogenies are determined by their kernels. Given a finite subgroup
$G \subset E(\Fpbar)$ that is stable under the Galois action, there is a curve $E/G$ and a
separable isogeny $\isog : E \to E/G$ with kernel $G$, both defined over $k$ and unique up to
post-composition with an isomorphism. V\'elu~\cite{Velu1971} gave explicit formulae for
$E/G$ and for $\isog$, which we now record in the general Weierstrass form
$E : y^2 + a_1xy + a_3y = x^3 + a_2x^2 + a_4x + a_6$ because that is the form in which they
are usually implemented~\cite{Shumow2009}.

Let $G_2 = \{ Q \in G : 2Q = \Ptinf, \ Q \neq \Ptinf \}$ and let $R$ contain exactly one point
from each pair $\{Q, -Q\}$ of $G \setminus (G_2 \cup \{\Ptinf\})$. Put $S = R \cup G_2$. For
$Q = (x_Q, y_Q) \in S$ define
\begin{align}
  g^x_Q &= 3x_Q^2 + 2a_2 x_Q + a_4 - a_1 y_Q, &
  g^y_Q &= -2y_Q - a_1 x_Q - a_3, \\
  v_Q &= \begin{cases}
           g^x_Q & \text{if } 2Q = \Ptinf, \\
           2g^x_Q - a_1 g^y_Q & \text{otherwise,}
         \end{cases} &
  u_Q &= (g^y_Q)^2 ,
\end{align}
and set $v = \sum_{Q \in S} v_Q$ and $w = \sum_{Q \in S} (u_Q + x_Q v_Q)$. Then
\begin{equation}
  E/G : \ y^2 + a_1 xy + a_3 y = x^3 + a_2 x^2 + (a_4 - 5v) x + \bigl(a_6 - (a_1^2 + 4a_2)v - 7w\bigr) ,
  \label{eq:velu-curve}
\end{equation}
and for a point $P = (x,y) \in E(\Fpbar)$ \emph{with $P \notin G$} the image $\isog(P)$ has
coordinates
\begin{align}
  X &= x + \sum_{Q \in S} \left( \frac{v_Q}{x - x_Q} - \frac{u_Q}{(x-x_Q)^2} \right),
  \label{eq:velu-x}\\
  Y &= y - \sum_{Q \in S} \left( u_Q \frac{2y + a_1 x + a_3}{(x-x_Q)^3}
      + v_Q \frac{a_1(x-x_Q) + y - y_Q}{(x-x_Q)^2}
      + \frac{a_1 u_Q - g^x_Q g^y_Q}{(x-x_Q)^2} \right).
  \label{eq:velu-y}
\end{align}

Two properties of these formulae govern how they may be used, and both are easy to overlook.

\begin{remark}[Domain of validity]
\label{rem:velu-domain}
\Cref{eq:velu-x} and \Cref{eq:velu-y} have poles at every $x = x_Q$ with $Q \in S$. They
compute $\isog(P)$ only for $P \notin G$; if $P$ generates $G$, then $\isog(P) = \Ptinf$ and
the formulae do not apply. An implementation that evaluates them at a kernel point and relies
on a modular inversion routine returning $0$ for the input $0$ will produce a well formed but
meaningless field element rather than an error.
\end{remark}

\begin{remark}[Cost]
\label{rem:velu-cost}
Evaluating the sums requires enumerating $S$, so the cost of V\'elu's formulae is
$\Theta(\deg \isog)$ field operations. They are therefore usable only for kernels of small
order. For a point $P$ of large order on a curve over a large field, the subgroup
$\langle P \rangle = \{P, 2P, \dots, \Ptinf\}$ cannot be enumerated at all, and any
construction that requires this enumeration does not scale beyond toy parameters. The square
root V\'elu algorithm of Bernstein, De Feo, Leroux and Smith~\cite{SqrtVelu2020} reduces the
cost to $\tilde{O}(\sqrt{\deg \isog})$ and moves the practical threshold upward, but it does
not remove the dependence on the degree. Large degree isogenies are computed as compositions
of many small degree steps, never in one application of \Cref{eq:velu-curve}.
\end{remark}

\subsection{Endomorphism rings and the class group action}
\label{sec:endo}

Let $E$ be a supersingular elliptic curve defined over $\Fp$. Two endomorphism rings must be
distinguished.

Over the algebraic closure, $\End_{\Fpbar}(E)$ is a maximal order in the quaternion algebra
$B_{p,\infty}$ ramified at $p$ and $\infty$~\cite{Deuring1941}. This ring is noncommutative,
and it is the object underlying SIDH and SQISign.

Over the prime field, the situation is different and much simpler. Write $\End_{\Fp}(E)$ for
the ring of endomorphisms of $E$ that are defined over $\Fp$.

\begin{proposition}[$\Fp$-rational endomorphisms are commutative]
\label{prop:endo-commutative}
Let $p > 3$ and let $E$ be a supersingular elliptic curve defined over $\Fp$, with Frobenius
endomorphism $\frob$. Then $\frob^2 = -p$, and $\End_{\Fp}(E)$ is an order $\Ord$ in the
imaginary quadratic field $\QQ(\sqrt{-p})$ containing $\ZZ[\frob]$. In particular
$\End_{\Fp}(E)$ is commutative.
\end{proposition}

\begin{proof}
By \Cref{prop:supersingular-equiv} the trace of Frobenius is zero, so $\frob$ satisfies
$\frob^2 + p = 0$ and $\ZZ[\frob] \cong \ZZ[\sqrt{-p}]$. An $\Fp$-rational endomorphism
commutes with $\frob$, so $\End_{\Fp}(E)$ is contained in the centraliser of $\frob$ inside
$\End_{\Fpbar}(E)$, which is a commutative ring of rank two over $\ZZ$ contained in
$\QQ(\frob) = \QQ(\sqrt{-p})$. The result is due to Waterhouse~\cite{Waterhouse1969}; see also
Delfs and Galbraith~\cite{DelfsGalbraith2016}.
\end{proof}

\begin{remark}[On the phrase ``isogenies computed over the rationals'']
\label{rem:rationals}
\Cref{prop:endo-commutative} is the precise statement behind the informal claim that
restricting to the prime field makes isogeny computation ``rational'' and therefore
commutative. The isogenies themselves are of course computed over $\Fp$, not over $\QQ$. What
lives in an imaginary quadratic field, and hence in a commutative ring, is the endomorphism
ring, and it is that commutativity, not any property of the arithmetic, which the protocols
of this family exploit. Curves over $\Fpsq$ do not have this property, since there
$\End(E)$ is a quaternion order.
\end{remark}

Fix an order $\Ord \subseteq \QQ(\sqrt{-p})$ containing $\ZZ[\frob]$ and let
\[
  \Ellset = \left\{ E/\Fp \ \text{supersingular} \ : \ \End_{\Fp}(E) \cong \Ord \right\} / \cong_{\Fp}
\]
be the set of $\Fp$-isomorphism classes of supersingular curves over $\Fp$ whose
$\Fp$-rational endomorphism ring is $\Ord$, with $\frob$ corresponding to the Frobenius
endomorphism under that identification. For an invertible ideal $\ida \subseteq \Ord$ put
$E[\ida] = \bigcap_{\alpha \in \ida} \ker \alpha$; this is a finite subgroup of $E$, and the
quotient isogeny $E \to E/E[\ida]$ has degree $\Norm(\ida)$. Writing $[\ida] \act E$ for the
class of $E/E[\ida]$ gives the following.

\begin{theorem}[Class group action]
\label{thm:class-action}
The map
\[
  \clgp \times \Ellset \longrightarrow \Ellset, \qquad ([\ida], E) \longmapsto [\ida] \act E
\]
is a well defined group action of the abelian group $\clgp$ on $\Ellset$, and it is free and
transitive.
\end{theorem}

\begin{proof}
See Waterhouse~\cite{Waterhouse1969}; the statement in this form and for this
choice of $\Ord$ is proved in~\cite{CSIDH2018}, building on the theory of complex
multiplication and on Deuring's correspondence~\cite{Deuring1941}. Freeness and transitivity
say that $\Ellset$ is a principal homogeneous space, or torsor, under $\clgp$.
\end{proof}

\Cref{fig:action} illustrates the structure that \Cref{thm:class-action} describes. Because the
action is free and transitive, the set $\Ellset$ is a copy of $\clgp$ with the origin
forgotten, and the isogeny graphs attached to the individual generators overlay one another on
that set. Commutativity is visible in the picture as the closing of every square built from
one step in each of two directions.

\begin{figure}[t]
\centering
\begin{tikzpicture}[
  font=\small,
  dot/.style={circle, fill=black, minimum size=3.6pt, inner sep=0pt},
  h/.style={-{Latex[length=1.7mm]}, gray!65, line width=0.5pt},
  v/.style={-{Latex[length=1.7mm]}, gray!65, line width=0.5pt, densely dashed},
  hb/.style={-{Latex[length=2.2mm]}, black, line width=1.0pt},
  vb/.style={-{Latex[length=2.2mm]}, black, line width=1.0pt, densely dashed},
]
  \foreach \i in {0,1,2,3} {
    \foreach \j in {0,1,2} {
      \node[dot] (n\i\j) at (2.9*\i, -1.65*\j) {};
    }
  }
  \foreach \j in {0,1,2} {
    \foreach \i/\k in {0/1, 1/2, 2/3} {
      \draw[h] (n\i\j) -- (n\k\j);
    }
    \draw[h] (n3\j) -- (10.35, -1.65*\j);
    \draw[h] (-0.75, -1.65*\j) -- (n0\j);
  }
  \foreach \i in {0,1,2,3} {
    \foreach \j/\l in {0/1, 1/2} {
      \draw[v] (n\i\j) -- (n\i\l);
    }
    \draw[v] (n\i2) -- (2.9*\i, -3.85);
    \draw[v] (2.9*\i, 0.55) -- (n\i0);
  }
  \draw[hb] (n00) -- node[above, font=\footnotesize] {$[\idl_1]$} (n10);
  \draw[vb] (n10) -- node[right, font=\footnotesize] {$[\idl_2]$} (n11);
  \draw[vb] (n00) -- node[left, font=\footnotesize] {$[\idl_2]$} (n01);
  \draw[hb] (n01) -- node[below, font=\footnotesize] {$[\idl_1]$} (n11);

  \node[anchor=south east, font=\footnotesize] at (-0.06,0.06) {$\Ezero$};
  \node[anchor=south west, font=\footnotesize] at (2.96,0.06) {$[\idl_1] \act \Ezero$};
  \node[anchor=north east, font=\footnotesize] at (-0.06,-1.71) {$[\idl_2] \act \Ezero$};
  \node[anchor=north west, font=\footnotesize] at (2.96,-1.71) {$[\idl_1\idl_2] \act \Ezero$};

  \node[anchor=west, font=\footnotesize] at (11.0,-0.55) {\ };
\end{tikzpicture}
\caption{The action of two generators on $\Ellset$. Each dot is an $\Fp$-isomorphism class of
supersingular curves; solid arrows are the action of $[\idl_1]$, that is $\ell_1$-isogenies,
and dashed arrows the action of $[\idl_2]$. The edges of each kind form disjoint cycles, and
the stubs at the boundary indicate that the picture wraps around. The highlighted square
commutes, that is $[\idl_1\idl_2] \act \Ezero = [\idl_2\idl_1] \act \Ezero$, which is
\Cref{eq:commutativity} for the two generators; it is this closing of squares that makes a
Diffie--Hellman style exchange possible.}
\label{fig:action}
\end{figure}
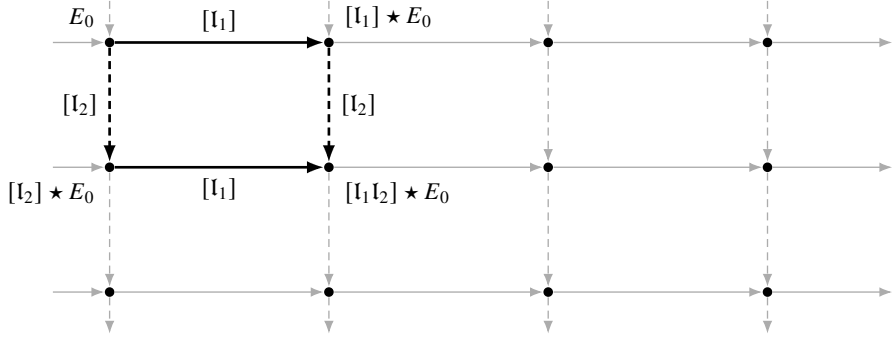

Two consequences of \Cref{thm:class-action} are used later. Since $\clgp$ is abelian,
\begin{equation}
  [\ida] \act \bigl( [\idb] \act E \bigr) \ = \ [\ida \idb] \act E \ = \
  [\idb] \act \bigl( [\ida] \act E \bigr)
  \label{eq:commutativity}
\end{equation}
for all classes $[\ida], [\idb]$ and all $E \in \Ellset$. Since the action is free and
transitive, $\#\Ellset = \#\clgp$, and each orbit element is reached by exactly one class.

Finally we record the fact that makes public keys short.

\begin{proposition}[Montgomery normal form]
\label{prop:montgomery-normal-form}
Let $p \equiv 3 \pmod 8$ and $p > 3$, and let $\Ord = \ZZ[\frob]$. Every class in
$\Ellset$ has a unique representative of the form $\Mont{A} : y^2 = x^3 + Ax^2 + x$ with
$A \in \Fp$. Consequently an element of $\Ellset$ is described by a single element of $\Fp$.
\end{proposition}

\begin{proof}
This is proved in~\cite{CSIDH2018}.
\end{proof}

\subsection{Effective evaluation of the action}
\label{sec:effective-action}

\Cref{thm:class-action} is not by itself an algorithm. To evaluate $[\ida] \act E$ one needs a
representation of $[\ida]$ for which the corresponding isogeny is computable. The standard
device, introduced in this setting by Couveignes~\cite{Couveignes2006} and Rostovtsev and
Stolbunov~\cite{RostovtsevStolbunov2006} for ordinary curves and by Castryck, Lange,
Martindale, Panny and Renes~\cite{CSIDH2018} for supersingular curves over $\Fp$, is to use
only ideals of small prime norm.

Suppose $\ell$ is an odd prime with $\ell \mid p+1$ and $\ell \nmid \disc(\Ord)$. Then
$\ell \Ord = \idl \bar{\idl}$ splits, with
\[
  \idl = (\ell, \frob - 1), \qquad \bar{\idl} = (\ell, \frob + 1) .
\]
The kernel $E[\idl] = \Etors{\ell} \cap \ker(\frob - 1)$ is exactly the subgroup of
$\Fp$-rational points of order $\ell$, so it is generated by a point that can be found by
sampling a random $P \in E(\Fp)$ and multiplying by the cofactor $(p+1)/\ell$. The isogeny is
then evaluated with V\'elu's formulae at cost $\Theta(\ell)$, which is acceptable because
$\ell$ is small. The conjugate ideal $\bar{\idl}$ is handled the same way using the quadratic
twist, whose $\Fp$-rational points are the points of $E$ with $x$-coordinate in $\Fp$ and
$y$-coordinate in $\Fpsq \setminus \Fp$.

Writing $\idl_1, \dots, \idl_n$ for the ideals attached to the primes $\ell_1, \dots, \ell_n$
dividing $p+1$, a private key is an exponent vector $\bm{e} = (e_1, \dots, e_n) \in \ZZ^n$
representing the class of $\ida = \prod_i \idl_i^{e_i}$, and evaluating the action costs
$\sum_i |e_i|$ small degree isogeny steps. This is a \emph{restricted} effective group action
in the terminology of Alamati, De Feo, Montgomery and Patranabis~\cite{AlamatiEtAl2020}: the
action of the generators is efficiently computable, but the group structure of $\clgp$ is not
known a priori, so one cannot sample uniformly from $\clgp$ or compute with arbitrary group
elements. Beullens, Kleinjung and Vercauteren~\cite{CSIFiSh2019} computed the class group and
a short relation lattice for the CSIDH-512 parameter set, which lifts the restriction for
that one parameter set and gives a genuine effective group action there.

\begin{definition}[Effective group action]
\label{def:ega}
An \emph{effective group action} is a triple $(\Gact, \Xset, \act)$ where $\Gact$ is a finite
abelian group, $\Xset$ a finite set, and $\act : \Gact \times \Xset \to \Xset$ a free and
transitive action, together with efficient algorithms for the group operations and equality
test in $\Gact$, for sampling from a distribution on $\Gact$ statistically close to uniform,
for a unique representation of elements of $\Xset$, and for computing $g \act x$ given
$g \in \Gact$ and $x \in \Xset$. A distinguished element $x_0 \in \Xset$ is called the origin.
\end{definition}

We will state the protocol and its security in the language of
\Cref{def:ega} so that the argument does not depend on the particular instantiation, and then
instantiate with $\Gact = \clgp$, $\Xset = \Ellset$ and $x_0 = \Ezero$ in
\Cref{sec:parameters}.

\subsection{Hard problems}
\label{sec:hard-problems}

Let $(\Gact, \Xset, \act)$ be an effective group action with origin $x_0$.

\begin{problem}[Vectorisation, \ensuremath{\mathrm{GA-VP}}]
\label{prob:vectorisation}
Given $x, y \in \Xset$, find $g \in \Gact$ with $g \act x = y$.
\end{problem}

\begin{problem}[Parallelisation, \ensuremath{\mathrm{GA-CDH}}]
\label{prob:cdh}
Given $x_0$, $g \act x_0$ and $h \act x_0$ for $g, h \getsr \Gact$, compute
$(gh) \act x_0$.
\end{problem}

\begin{problem}[Decisional parallelisation, \ensuremath{\mathrm{GA-DDH}}]
\label{prob:ddh}
Distinguish the distributions $(g \act x_0, h \act x_0, (gh) \act x_0)$ and
$(g \act x_0, h \act x_0, r \act x_0)$ for $g,h,r \getsr \Gact$.
\end{problem}

Vectorisation is the group action analogue of the discrete logarithm problem and
parallelisation the analogue of computational Diffie--Hellman. Note that neither problem is a
discrete logarithm problem: there is no exponentiation map and no pairing, and in particular
the group $\Gact$ is not given to the adversary as a set of exponents acting on a cyclic
group. We return to the consequences of this in \Cref{sec:pohlig}.

The decisional problem requires care, and we do not assume it.

\begin{remark}[\ensuremath{\mathrm{GA-DDH}} is false for the instantiation we use]
\label{rem:ddh-broken}
Castryck, Sot\'akov\'a and Vercauteren~\cite{CSV2020} showed that for class group actions
whose discriminant is not prime, genus theory supplies efficiently computable quadratic
characters on $\clgp$ that are also computable from the acted-upon curve, and that these
characters distinguish the two distributions of \Cref{prob:ddh}. The CSIDH instantiation has
$\Ord = \ZZ[\sqrt{-p}]$ of discriminant $-4p$, which is not prime, so \ensuremath{\mathrm{GA-DDH}} is
broken there. Any security argument for a protocol in this family must therefore avoid
\ensuremath{\mathrm{GA-DDH}}. We do so by working with hashed values in the random oracle model and
reducing to the following computational problem instead.
\end{remark}

\begin{problem}[Strong parallelisation, \ensuremath{\mathrm{GA-StCDH}}]
\label{prob:stcdh}
Given $x_0$, $g \act x_0$ and $h \act x_0$ for $g,h \getsr \Gact$, compute $(gh) \act x_0$,
with the help of an oracle $\mathcal{D}_{g}(\cdot,\cdot)$ that on input $(u, w) \in \Xset^2$
returns $1$ if and only if $w = g \act u$.
\end{problem}

The oracle in \Cref{prob:stcdh} fixes its first argument to the challenge element $g$, which
is exactly the form needed to make random oracle proofs of Diffie--Hellman style key
agreement go through; it plays the same role as the oracle in the oracle Diffie--Hellman
assumption of Abdalla, Bellare and Rogaway~\cite{ABR2001}. Strong variants of group
action problems in this shape appear in the isogeny based key exchange
literature~\cite{deKock2021,Kawashima2021}.

\begin{assumption}[\ensuremath{\mathrm{GA-StCDH}}]
\label{ass:stcdh}
For the parameter sets of \Cref{sec:parameters}, no probabilistic polynomial time algorithm
solves \Cref{prob:stcdh} with non-negligible probability. We write
$\advantage{StCDH}{\Gact,\Xset}(\Bdv)$ for the success probability of an algorithm $\Bdv$
against \Cref{prob:stcdh}.
\end{assumption}

\Cref{ass:stcdh} is the only computational assumption about the group action that our results
use. \Cref{sec:attacks} surveys what is known about its difficulty, and
\Cref{sec:security-levels} discusses the substantial disagreement in the literature about the
concrete quantum security of the smallest parameter set.

\subsection{Authenticated encryption}
\label{sec:aead}

The masking layer of our protocol is an authenticated encryption scheme with associated data,
which we treat as a black box. Such a scheme
$\Pi = (\AEnc, \ADec)$ takes a key $K \in \bits^{\secpar}$, a nonce $r$, associated data
$\mathit{ad}$ and a message $M$, and produces a ciphertext $C = \AEnc_K(r, \mathit{ad}, M)$
that $\ADec_K(r, \mathit{ad}, C)$ either maps back to $M$ or rejects. We use the standard
notions of indistinguishability under chosen plaintext attack and integrity of ciphertexts,
written $\advantage{ind\text{-}cpa}{\Pi}(\Adv)$ and $\advantage{int\text{-}ctxt}{\Pi}(\Adv)$;
together they imply security under chosen ciphertext attack by the result of Bellare and
Namprempre~\cite{BellareNamprempre2008}, and the associated data formalism is that of
Rogaway~\cite{Rogaway2002}.
\section{The \MCSI{} protocol}
\label{sec:protocol}

\subsection{Design rationale}
\label{sec:rationale}

Before specifying the protocol we explain the three design decisions that shape it, because
each of them rules out a construction that looks natural at first sight.

\paragraph{The secret must act through a group action, not through an iterated walk.}
A tempting way to build a Diffie--Hellman analogue from isogenies is to define one ``step''
as an isogeny whose kernel is generated by the current base point, to push the base point
through that isogeny, and to iterate the step $k$ times, taking $k$ as the private key. This
does not work, for two independent reasons. First, the base point generates the kernel, so its
image under the step isogeny is the point at infinity and the iteration is undefined; see
\Cref{rem:velu-domain}. Second, even after repairing this by pushing some other point, the
resulting map depends on the whole trajectory rather than on an abstract group element, and
the two orders of composition do not agree: if $\Phi_k$ denotes the $k$-fold iteration
starting from $E$, then in general
$\Phi_{k}^{(\Phi_{j}(E))} \circ \Phi_{j}^{(E)} \neq \Phi_{j}^{(\Phi_{k}(E))} \circ \Phi_{k}^{(E)}$,
because the intermediate base points differ on the two sides. Commutativity is not an
incidental property that can be hoped for; it is exactly the statement of
\Cref{eq:commutativity}, and it holds because $\clgp$ is an abelian group acting on a torsor,
not because isogenies are in any sense commutative. We therefore build the protocol on the
class group action of \Cref{thm:class-action} and on nothing else.

\paragraph{A blinding key cannot be derived from the value the blinded message is meant to establish.}
Suppose the first protocol message carries a blinded payload, and suppose the blinding key is
derived from the session secret. The receiver can derive that key only after it has obtained
the session secret, and it obtains the session secret only by processing the message. The
requirement is therefore circular and no ordering of the steps resolves it. The blinding key
must be computable by both parties from material that exists before the session starts, and
there are only three kinds of such material: public data, which provides no confidentiality;
a pre-shared symmetric secret, which defeats the purpose of a public key protocol; or a
secret determined by long term asymmetric keys. We take the third option. Each party holds a
static key pair, the two static keys determine a shared value
$Z_{\mathrm{ss}} = (ab) \act x_0$ by \Cref{eq:commutativity}, and this value, which can be
computed once per peer and cached, keys the blinding layer.

This choice has a consequence that must be stated plainly rather than glossed over: it turns
\MCSI{} into an implicitly authenticated key exchange with static keys, and it presupposes an
authentic distribution of static public keys. It also means that the blinding is not forward
secret, since an adversary who later learns $a$ or $b$ can recompute
$Z_{\mathrm{ss}}$ and remove the blinding from recorded transcripts. The session key does
remain secure in that situation, because it additionally depends on the ephemeral secrets;
see \Cref{sec:limits}. The pattern is the group action analogue of the static and ephemeral
mixing used in Diffie--Hellman based authenticated key exchange, and the accounting of which
secrets protect which property is the same.

\paragraph{A public bijection is not a security layer, and a keystream without a tag is not integrity.}
It is sometimes proposed to strengthen an exclusive-or blinding step by pushing every byte of
the result through a fixed, publicly known substitution box. This adds nothing at all, and
the statement can be made exact.

\begin{lemma}[Public bijections do not change advantage]
\label{lem:sbox}
Let $\Pi$ be an encryption scheme with ciphertext space $\bits^{8n}$ and let
$S : \bits^{8} \to \bits^{8}$ be a bijection whose description is public. Let $\Pi^S$ be the
scheme obtained from $\Pi$ by applying $S$ to each byte of every ciphertext, and applying
$S^{-1}$ before every decryption. Then for every notion of security defined by a game in
which the adversary's interface consists of encryption and decryption oracles, and for every
adversary $\Adv$ against $\Pi^S$, there is an adversary $\Bdv$ against $\Pi$ with
$\advantage{}{\Pi}(\Bdv) = \advantage{}{\Pi^S}(\Adv)$ and whose running time exceeds that of
$\Adv$ by at most $n$ table lookups per oracle call.
\end{lemma}

\begin{proof}
$\Bdv$ runs $\Adv$ and relays its oracle queries, applying $S$ bytewise to every ciphertext
it passes to $\Adv$ and $S^{-1}$ bytewise to every ciphertext it receives from $\Adv$. Since
$S$ is a bijection with a public description, the view of $\Adv$ inside this simulation is
distributed exactly as in the real game against $\Pi^S$, and $\Bdv$ outputs whatever $\Adv$
outputs. The same construction in the other direction gives the converse.
\end{proof}

We therefore include no such layer. For the same reason we do not use a bare keystream: an
exclusive-or of the payload with a pseudorandom string is malleable, and an active adversary
who flips bits of the ciphertext flips the corresponding bits of the payload. Since our
payload is an element of $\Xset$, and since a modified element of $\Xset$ is exactly what
adaptive attacks against static key isogeny protocols
require~\cite{GPST2016}, malleability here is not a theoretical concern. The masking layer is
consequently a full authenticated encryption scheme, and the receiver additionally validates
the recovered element of $\Xset$ before acting on it.

\subsection{Public parameters and building blocks}
\label{sec:params-syntax}

\MCSI{} is parameterised by an effective group action $(\Gact, \Xset, \act)$ in the sense of
\Cref{def:ega} with origin $x_0$, a hash function
$\Hash : \bits^{*} \to \bits^{2\secpar}$ modelled as a random oracle, an authenticated
encryption scheme $\Pi = (\AEnc, \ADec)$ with $\secpar$-bit keys, and an encoding
$\encode : \Xset \to \bits^{8L}$ with inverse $\decode$ defined on the image. The parameters
also fix a protocol label $\mathrm{lbl} = \text{\texttt{"MCSI-v1"}}$ and a nonce length
$\nu$. We write $\Validate(\cdot)$ for the algorithm that decides membership in $\Xset$; its
instantiation is given in \Cref{sec:validation}.

The concrete instantiation used in this paper is
$\Gact = \clgp$, $\Xset = \Ellset$ and $x_0 = \Ezero : y^2 = x^3 + x$, with the parameter sets
of \Cref{sec:parameters}. There $L = \lceil \log_2 p / 8 \rceil$ and $\encode$ is the
big-endian encoding of the Montgomery coefficient supplied by
\Cref{prop:montgomery-normal-form}.

\subsection{Static keys}
\label{sec:static-keys}

Static key generation, given as \Cref{alg:keygen}, is the group action key generation of
Couveignes~\cite{Couveignes2006}, Rostovtsev and Stolbunov~\cite{RostovtsevStolbunov2006} and
Castryck et al.~\cite{CSIDH2018}.

\begin{algorithm}[t]
\caption{$\MCSI.\KeyGen$}
\label{alg:keygen}
\begin{algorithmic}[1]
  \State $a \getsr \Gact$ \Comment{for a restricted action, sample an exponent vector as in \Cref{sec:effective-action}}
  \State $\mathit{pk} \gets a \act x_0$
  \State \Return $(a, \mathit{pk})$
\end{algorithmic}
\end{algorithm}

Static public keys are distributed authentically, for example through a certificate or an
out-of-band channel. On receipt of a static public key a party runs $\Validate$ on it once
and rejects it if it fails. This is not optional: an adaptive attacker who is able to submit
malformed public keys to a party that reuses a static secret can recover that secret one bit
at a time, which is the group action analogue of the attack of Galbraith, Petit, Shani and
Ti~\cite{GPST2016} and is discussed for CSIDH in~\cite{CSIDH2018}.

\subsection{The masking layer}
\label{sec:masking}

Both parties derive two directional masking keys from the static-static shared value together
with the two static public keys. Since either party may initiate, the derivation must not
depend on who does, so we order the two public keys canonically. Let $P_1, P_2$ be the two
static public keys of the pair, labelled so that $\encode(P_1) \leq \encode(P_2)$ in
lexicographic order on byte strings, and put
\begin{equation}
  Z_{\mathrm{ss}} \gets a \act \mathit{pk}_B = b \act \mathit{pk}_A , \qquad
  \bigl(K_{1 \to 2},\, K_{2 \to 1}\bigr) \gets
  \Hash\bigl( \mathrm{lbl} \cat \texttt{"mask"} \cat \encode(Z_{\mathrm{ss}})
              \cat \encode(P_1) \cat \encode(P_2) \bigr) .
  \label{eq:mask-keys}
\end{equation}
Each party uses the half that matches the direction it is sending in. We write $K_{A \to B}$
for whichever of $K_{1 \to 2}, K_{2 \to 1}$ carries messages from Alice to Bob, and
$K_{B \to A}$ for the other. The two directions therefore never share a key, so nonce
collisions between directions are impossible and a message cannot be reflected back at its
sender. The value $Z_{\mathrm{ss}}$ and the pair of keys depend only on the two static key
pairs and not on any session or on any choice of roles, so they are computed once per peer and
cached; no per-session evaluation of the group action is spent on them.

Masking and unmasking are then
\begin{equation}
  \Mask_K(r, \mathit{ad}, x) = \AEnc_K\bigl(r, \mathit{ad}, \encode(x)\bigr), \qquad
  \Unmask_K(r, \mathit{ad}, c) = \decode\bigl(\ADec_K(r, \mathit{ad}, c)\bigr) ,
\end{equation}
where $\Unmask$ returns $\bot$ if $\ADec$ rejects or if the recovered string is not a valid
encoding.

We recommend instantiating $\Pi$ with a standard authenticated encryption scheme such as
AES-256-GCM or ChaCha20-Poly1305. For deployments that prefer to depend on a single hash
primitive, the following hash based instantiation is also adequate and is closer to the
counter mode construction from which this design grew. Let
$Z = Z_0 \cat Z_1 \cat \cdots$ with $Z_i = \Hash(K_e \cat r \cat \Itos(i,4))$ truncated to
$8L$ bits, set $c_0 = \encode(x) \oplus Z$, and set
$\tau = \Mac_{K_m}(r \cat \mathit{ad} \cat c_0)$ with $\Mac$ instantiated by HMAC. The
ciphertext is $c = (c_0, \tau)$ and the keys $K_e, K_m$ are the two halves of a hash of $K$.
This is encrypt-then-MAC, so its authenticated encryption security follows from the
pseudorandomness of the keystream and the unforgeability of the MAC by the composition
theorem of Bellare and Namprempre~\cite{BellareNamprempre2008}.

\begin{remark}[The payload is redundant, and this is harmless]
\label{rem:redundancy}
For the instantiation of \Cref{sec:parameters} the payload is an element of $\Xset$ with
$\#\Xset \approx \sqrt{p}$, encoded into $L = \lceil \log_2 p / 8 \rceil$ bytes. The encoding
is therefore highly redundant: only about a $p^{-1/2}$ fraction of byte strings of that length
decode to valid elements, and $\Validate$ detects the rest. A consequence is that a guessed
masking key can be tested offline by checking whether the recovered string is a valid element
of $\Xset$. This does not weaken the scheme, since the key space has $2^{\secpar}$ elements
and the test only replaces one constant factor by another, but it does mean that the
ciphertext is not indistinguishable from random to a party that knows the key, and any claim
of key privacy would have to be argued separately. We make no such claim.
\end{remark}

\subsection{The protocol}
\label{sec:the-protocol}

\Cref{fig:protocol} gives the protocol in full. Alice and Bob hold static key pairs
$(a, \mathit{pk}_A)$ and $(b, \mathit{pk}_B)$ and have already validated each other's static
public keys and cached the masking keys of \Cref{eq:mask-keys}. The session consists of one
round trip.

\begin{figure}[t]
\centering
\begin{tikzpicture}[
  font=\footnotesize,
  msg/.style={-{Latex[length=2.2mm]}, thick},
]
  \def\lx{0}
  \def\rx{15.4}

  \node[anchor=west] at (\lx, 0.6) {\normalsize\textbf{Alice} \ $(a,\ \mathit{pk}_A = a \act x_0)$};
  \node[anchor=east] at (\rx, 0.6) {\normalsize\textbf{Bob} \ $(b,\ \mathit{pk}_B = b \act x_0)$};
  \draw[thick] (\lx,0.25) -- (\rx,0.25);

  \draw[gray!60] (\lx,0.15) -- (\lx,-9.55);
  \draw[gray!60] (\rx,0.15) -- (\rx,-9.55);

  \node[anchor=west] at (\lx+0.15,-0.35) {$Z_{\mathrm{ss}} \gets a \act \mathit{pk}_B$};
  \node[anchor=east] at (\rx-0.15,-0.35) {$Z_{\mathrm{ss}} \gets b \act \mathit{pk}_A$};
  \node[anchor=west] at (\lx+0.15,-0.95)
    {$(K_{1 \to 2},\, K_{2 \to 1}) \gets \Hash\bigl(\mathrm{lbl} \cat \texttt{"mask"} \cat Z_{\mathrm{ss}} \cat P_1 \cat P_2\bigr)$, \ with $\encode(P_1) \leq \encode(P_2)$};
  \node[anchor=east, gray] at (\rx-0.15,-0.95) {\itshape computed once per peer};
  \draw[gray!60, dashed] (\lx,-1.35) -- (\rx,-1.35);

  \node[anchor=west] at (\lx+0.15,-1.85) {$u \getsr \Gact$, \quad $T_A \gets u \act x_0$, \quad $r_A \getsr \bits^{\nu}$};
  \node[anchor=west] at (\lx+0.15,-2.4) {$\mathit{ad}_1 \gets \mathrm{lbl} \cat \mathit{pk}_A \cat \mathit{pk}_B \cat r_A$};
  \node[anchor=west] at (\lx+0.15,-2.95) {$c_A \gets \Mask_{K_{A \to B}}(r_A, \mathit{ad}_1, T_A)$};

  \draw[msg] (\lx,-3.55) -- node[above,font=\footnotesize] {$m_1 = (r_A,\ c_A)$} (\rx,-3.55);

  \node[anchor=east] at (\rx-0.15,-4.15) {$T_A \gets \Unmask_{K_{A \to B}}(r_A, \mathit{ad}_1, c_A)$; \ abort if $T_A = \bot$ or $\Validate(T_A) = 0$};
  \node[anchor=east] at (\rx-0.15,-4.7) {$w \getsr \Gact$, \quad $T_B \gets w \act x_0$, \quad $r_B \getsr \bits^{\nu}$};
  \node[anchor=east] at (\rx-0.15,-5.25) {$\mathit{ad}_2 \gets \mathrm{lbl} \cat \mathit{pk}_B \cat \mathit{pk}_A \cat r_B \cat m_1$};
  \node[anchor=east] at (\rx-0.15,-5.8) {$c_B \gets \Mask_{K_{B \to A}}(r_B, \mathit{ad}_2, T_B)$};

  \draw[msg] (\rx,-6.4) -- node[above,font=\footnotesize] {$m_2 = (r_B,\ c_B)$} (\lx,-6.4);

  \node[anchor=west] at (\lx+0.15,-7.0) {$T_B \gets \Unmask_{K_{B \to A}}(r_B, \mathit{ad}_2, c_B)$; \ abort if $T_B = \bot$ or $\Validate(T_B) = 0$};

  \draw[gray!60, dashed] (\lx,-7.4) -- (\rx,-7.4);

  \node[anchor=west] at (\lx+0.15,-7.9)
    {Alice: \ $Z_{\mathrm{ee}} \gets u \act T_B$, \qquad $Z_{\mathrm{es}} \gets u \act \mathit{pk}_B$, \qquad $Z_{\mathrm{se}} \gets a \act T_B$};
  \node[anchor=west] at (\lx+0.15,-8.45)
    {Bob: \ \ \ \ $Z_{\mathrm{ee}} \gets w \act T_A$, \qquad $Z_{\mathrm{es}} \gets b \act T_A$, \qquad $Z_{\mathrm{se}} \gets w \act \mathit{pk}_A$};
  \node[anchor=west] at (\lx+0.15,-9.0)
    {$\tr \gets \mathrm{lbl} \cat \mathit{pk}_A \cat \mathit{pk}_B \cat m_1 \cat m_2$, \qquad
     $\mathit{sk} \gets \Hash\bigl(\mathrm{lbl} \cat \texttt{"key"} \cat Z_{\mathrm{ee}} \cat Z_{\mathrm{es}} \cat Z_{\mathrm{se}} \cat Z_{\mathrm{ss}} \cat \tr \bigr)$};
\end{tikzpicture}
\caption{The \MCSI{} protocol. Elements of $\Xset$ are encoded with $\encode$ wherever they
appear as hash or associated-data inputs; the encoding is omitted from the figure for
readability. The two steps above the first dashed rule depend only on the static keys and are
computed once per peer.}
\label{fig:protocol}
\end{figure}
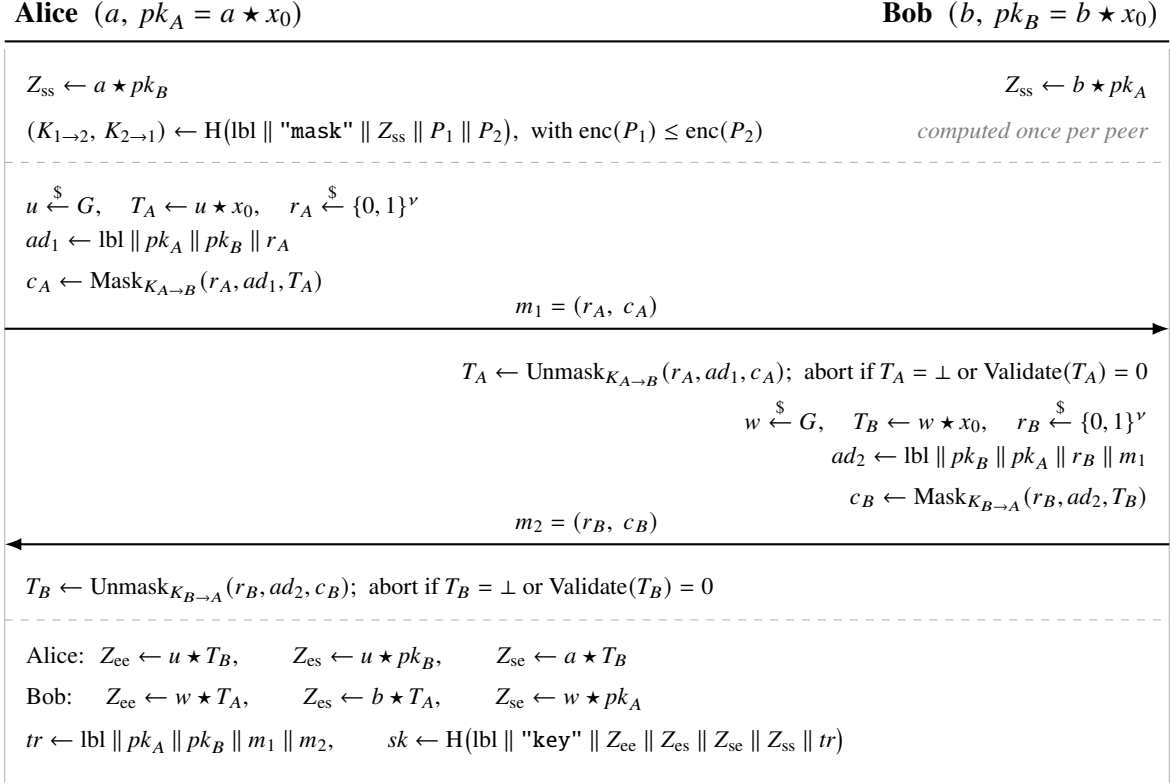

The three shared values other than $Z_{\mathrm{ss}}$ are
\begin{equation}
  Z_{\mathrm{ee}} = (uw) \act x_0, \qquad
  Z_{\mathrm{es}} = (ub) \act x_0, \qquad
  Z_{\mathrm{se}} = (aw) \act x_0 ,
  \label{eq:three-dh}
\end{equation}
and each is computed by the two parties along different routes;
\Cref{sec:correctness} proves that the routes agree. Their roles are distinct.
$Z_{\mathrm{ee}}$ depends only on ephemeral secrets and supplies forward secrecy for the
session key. $Z_{\mathrm{es}}$ and $Z_{\mathrm{se}}$ each mix one static secret with one
ephemeral secret and supply implicit authentication in one direction.
$Z_{\mathrm{ss}}$ depends only on static secrets; it is already needed for the masking layer,
so including it in the session key derivation costs nothing.

\Cref{alg:peerctx} and \Cref{alg:session} restate the same protocol as pseudocode, in the form
an implementation would take. \textsc{PeerCtx} is run once for each peer and its output is
cached; \textsc{Init}, \textsc{Resp} and \textsc{Fin} are the three per-session steps.

\begin{algorithm}[t]
\caption{Per-peer setup, run once and cached}
\label{alg:peerctx}
\begin{algorithmic}[1]
\Procedure{PeerCtx}{$\mathit{sk}_{\mathrm{self}},\ \mathit{pk}_{\mathrm{self}},\ \mathit{pk}_{\mathrm{peer}}$}
  \State \textbf{return} $\bot$ \textbf{if} $\Validate(\mathit{pk}_{\mathrm{peer}}) = 0$
  \State $Z_{\mathrm{ss}} \gets \mathit{sk}_{\mathrm{self}} \act \mathit{pk}_{\mathrm{peer}}$
  \State let $(P_1, P_2)$ be $\{\mathit{pk}_{\mathrm{self}}, \mathit{pk}_{\mathrm{peer}}\}$ ordered so that $\encode(P_1) \leq \encode(P_2)$
  \State $(K_{1 \to 2}, K_{2 \to 1}) \gets \Hash\bigl(\mathrm{lbl} \cat \texttt{"mask"} \cat \encode(Z_{\mathrm{ss}}) \cat \encode(P_1) \cat \encode(P_2)\bigr)$
  \State $K_{\mathrm{out}} \gets K_{1 \to 2}$ \textbf{if} $\mathit{pk}_{\mathrm{self}} = P_1$ \textbf{else} $K_{2 \to 1}$; \quad $K_{\mathrm{in}} \gets$ the other half
  \State \Return $\ctxt = (\mathit{sk}_{\mathrm{self}},\, Z_{\mathrm{ss}},\, K_{\mathrm{out}},\, K_{\mathrm{in}},\, \mathit{pk}_{\mathrm{self}},\, \mathit{pk}_{\mathrm{peer}})$
\EndProcedure
\end{algorithmic}
\end{algorithm}

\begin{algorithm}[t]
\caption{The per-session steps. Every procedure reads $\mathit{sk}_{\mathrm{self}}$,
$Z_{\mathrm{ss}}$, $K_{\mathrm{out}}$, $K_{\mathrm{in}}$, $\mathit{pk}_{\mathrm{self}}$ and
$\mathit{pk}_{\mathrm{peer}}$ from the cached $\ctxt$ of \Cref{alg:peerctx}. Here
$\mathit{pk}_A$ is the initiator's static public key and $\mathit{pk}_B$ the responder's, so
the initiator has $\mathit{pk}_A = \mathit{pk}_{\mathrm{self}}$ and the responder has
$\mathit{pk}_A = \mathit{pk}_{\mathrm{peer}}$.}
\label{alg:session}
\begin{algorithmic}[1]
\Procedure{Init}{$\ctxt$}
  \State $u \getsr \Gact$; \quad $T_A \gets u \act x_0$; \quad $r_A \getsr \bits^{\nu}$
  \State $\mathit{ad}_1 \gets \mathrm{lbl} \cat \encode(\mathit{pk}_A) \cat \encode(\mathit{pk}_B) \cat r_A$
  \State $c_A \gets \AEnc_{K_{\mathrm{out}}}\bigl(r_A,\, \mathit{ad}_1,\, \encode(T_A)\bigr)$; \quad $m_1 \gets (r_A, c_A)$
  \State \Return $\bigl(\mathit{st} = (u, m_1),\ m_1\bigr)$
\EndProcedure
\Statex
\Procedure{Resp}{$\ctxt, m_1$}
  \State parse $m_1$ as $(r_A, c_A)$, returning $\bot$ on failure
  \State $\mathit{ad}_1 \gets \mathrm{lbl} \cat \encode(\mathit{pk}_A) \cat \encode(\mathit{pk}_B) \cat r_A$
  \State $M \gets \ADec_{K_{\mathrm{in}}}(r_A, \mathit{ad}_1, c_A)$; \quad \textbf{return} $\bot$ \textbf{if} $M = \bot$
  \State $T_A \gets \decode(M)$; \quad \textbf{return} $\bot$ \textbf{if} $T_A = \bot$ \textbf{or} $\Validate(T_A) = 0$
  \State $w \getsr \Gact$; \quad $T_B \gets w \act x_0$; \quad $r_B \getsr \bits^{\nu}$
  \State $\mathit{ad}_2 \gets \mathrm{lbl} \cat \encode(\mathit{pk}_B) \cat \encode(\mathit{pk}_A) \cat r_B \cat m_1$
  \State $c_B \gets \AEnc_{K_{\mathrm{out}}}\bigl(r_B,\, \mathit{ad}_2,\, \encode(T_B)\bigr)$; \quad $m_2 \gets (r_B, c_B)$
  \State $Z_{\mathrm{ee}} \gets w \act T_A$; \quad $Z_{\mathrm{es}} \gets \mathit{sk}_{\mathrm{self}} \act T_A$; \quad $Z_{\mathrm{se}} \gets w \act \mathit{pk}_{\mathrm{peer}}$
  \State \Return $\bigl(\textsc{Derive}(\ctxt, Z_{\mathrm{ee}}, Z_{\mathrm{es}}, Z_{\mathrm{se}}, m_1, m_2),\ m_2\bigr)$
\EndProcedure
\Statex
\Procedure{Fin}{$\ctxt, \mathit{st}, m_2$}
  \State parse $\mathit{st}$ as $(u, m_1)$ and $m_2$ as $(r_B, c_B)$, returning $\bot$ on failure
  \State $\mathit{ad}_2 \gets \mathrm{lbl} \cat \encode(\mathit{pk}_B) \cat \encode(\mathit{pk}_A) \cat r_B \cat m_1$
  \State $M \gets \ADec_{K_{\mathrm{in}}}(r_B, \mathit{ad}_2, c_B)$; \quad \textbf{return} $\bot$ \textbf{if} $M = \bot$
  \State $T_B \gets \decode(M)$; \quad \textbf{return} $\bot$ \textbf{if} $T_B = \bot$ \textbf{or} $\Validate(T_B) = 0$
  \State $Z_{\mathrm{ee}} \gets u \act T_B$; \quad $Z_{\mathrm{es}} \gets u \act \mathit{pk}_{\mathrm{peer}}$; \quad $Z_{\mathrm{se}} \gets \mathit{sk}_{\mathrm{self}} \act T_B$
  \State \Return $\textsc{Derive}(\ctxt, Z_{\mathrm{ee}}, Z_{\mathrm{es}}, Z_{\mathrm{se}}, m_1, m_2)$
\EndProcedure
\Statex
\Procedure{Derive}{$\ctxt, Z_{\mathrm{ee}}, Z_{\mathrm{es}}, Z_{\mathrm{se}}, m_1, m_2$}
  \State $\tr \gets \mathrm{lbl} \cat \encode(\mathit{pk}_A) \cat \encode(\mathit{pk}_B) \cat m_1 \cat m_2$
  \State \Return $\Hash\bigl(\mathrm{lbl} \cat \texttt{"key"} \cat \encode(Z_{\mathrm{ee}}) \cat \encode(Z_{\mathrm{es}}) \cat \encode(Z_{\mathrm{se}}) \cat \encode(Z_{\mathrm{ss}}) \cat \tr\bigr)$
\EndProcedure
\end{algorithmic}
\end{algorithm}

\begin{remark}[Key registration, unknown key share and reflection]
\label{rem:registration}
Three points about how identities enter the protocol. First, the two static public keys enter
both the masking key derivation of \Cref{eq:mask-keys} and the transcript hashed into the
session key, so a completed session binds the key to the identities of both parties and an
unknown key share attack would have to make one party accept a different identity for the
same session, which the transcript prevents. Second, registration should nevertheless require
proof of possession of the static secret. A party that registers a copy of somebody else's
public key cannot run the protocol under it, since it cannot compute $Z_{\mathrm{ss}}$, but
allowing such registrations invites confusion in the surrounding system and costs nothing to
forbid. Third, the two directions use different keys by construction, so a message cannot be
replayed back at its sender as a reply; the associated data of the two messages also differ
in the order of the two public keys, which blocks the same attack a second time.
\end{remark}

\begin{remark}[Identifying the peer, and replay]
\label{rem:peer-id}
Two engineering points follow from the fact that the first message is encrypted. First, the
recipient must know which cached masking key to try, and the message as specified carries no
identifier. A deployment must either attach a sender or key identifier in the clear, which
reduces the privacy benefit of the blinding to the ephemeral element alone, or have the
recipient try its cached keys in turn, which is only practical when the number of peers is
small. We regard the first option as the normal one and note the cost rather than hiding it.
Second, an adversary can replay a recorded $m_1$ to the responder, who will run one session in
response. This does not produce a repeated session key, because the key binds the full
transcript and the responder contributes a fresh ephemeral element and a fresh nonce, and it
does not let the adversary learn the key, but it does let a party that has recorded one valid
message make the responder do work. A responder that cares about this should cache recently
seen nonces and reject repeats.
\end{remark}

\subsection{Optional explicit key confirmation}
\label{sec:key-confirmation}

As specified, \MCSI{} provides implicit authentication only: a party that completes the
protocol knows that nobody other than the intended peer can compute the session key, but it
does not know that the peer actually completed the protocol. Where explicit confirmation is
wanted, derive an additional confirmation key
$K_c$ from the same hash call by taking a longer output, and append
$\tau_A = \Mac_{K_c}(\texttt{"A"} \cat \tr)$ to a third message and
$\tau_B = \Mac_{K_c}(\texttt{"B"} \cat \tr)$ to a fourth. This adds a second round trip and
does not change any of the statements proved in \Cref{sec:correctness} and
\Cref{sec:security}, all of which concern the two message protocol.

\subsection{Cost}
\label{sec:cost}

Each party evaluates the group action four times per session, namely once to produce its own
ephemeral element and three times to produce
$Z_{\mathrm{ee}}, Z_{\mathrm{es}}, Z_{\mathrm{se}}$, plus one further evaluation per peer for
$Z_{\mathrm{ss}}$ that is cached across all sessions with that peer. Unauthenticated
Diffie--Hellman over the same group action costs two evaluations per party, so \MCSI{} pays a
factor of two in group action evaluations for implicit mutual authentication and for blinding
of the ephemeral element. The symmetric cost is two hash calls and two authenticated
encryption operations per party, which is negligible beside a single evaluation of the action.

On the wire, each message carries a nonce, an $L$-byte ciphertext and an authentication tag.
With $\nu = 128$, a $128$-bit tag and the CSIDH-512 instantiation, where $L = 64$, each
message is $96$ bytes and the whole exchange is $192$ bytes. These figures follow from the
parameter choices rather than from measurement, and the implementation of \Cref{sec:impl}
produces exactly them. What the same implementation costs in time is reported in
\Cref{sec:measured-cost}.
\section{Correctness}
\label{sec:correctness}

Correctness of \MCSI{} is not an assumption. It follows from the commutativity of $\Gact$
together with the fact that elements of $\Xset$ have a unique representation, and we prove it
here. The proof is short, but it is worth writing out because it identifies exactly which
algebraic property the protocol depends on, and therefore what would have to be re-established
if the group action were replaced by something else.

\begin{theorem}[Correctness]
\label{thm:correctness}
Let $(\Gact, \Xset, \act)$ be an effective group action in the sense of \Cref{def:ega}, let
$\Pi$ be a perfectly correct authenticated encryption scheme, and let $\encode$ be injective
on $\Xset$. Suppose Alice and Bob execute \MCSI{} of \Cref{fig:protocol} with honestly
generated static keys, and suppose both messages are delivered unmodified. Then neither party
aborts, and both compute the same session key.
\end{theorem}

\begin{proof}
Write $a, b \in \Gact$ for the static secrets and $u, w \in \Gact$ for the ephemeral secrets,
so that $\mathit{pk}_A = a \act x_0$, $\mathit{pk}_B = b \act x_0$, $T_A = u \act x_0$ and
$T_B = w \act x_0$.

\emph{Step 1: the masking keys agree.} Alice computes
$a \act \mathit{pk}_B = a \act (b \act x_0) = (ab) \act x_0$, using that $\act$ is a group
action. Bob computes $b \act \mathit{pk}_A = b \act (a \act x_0) = (ba) \act x_0$. Since
$\Gact$ is abelian, $ab = ba$, so the two parties obtain the same element of $\Xset$. By
\Cref{def:ega} elements of $\Xset$ have a unique representation, so the two parties obtain the
same bit string $\encode(Z_{\mathrm{ss}})$. The remaining inputs to \Cref{eq:mask-keys} are the
two static public keys in the canonical order, which both parties can compute from the pair
they hold, so the hash inputs coincide and hence so do the two directional keys. Each party
then selects the half matching the direction it is sending in, and since the two parties
disagree about neither the pair nor the direction, they select consistently.

\emph{Step 2: the ephemeral elements are recovered.} Bob computes $\mathit{ad}_1$ from
$\mathrm{lbl}$, the two static public keys and the nonce $r_A$ carried in $m_1$, so his
associated data equals Alice's. By Step 1 he holds the key Alice encrypted under, and $m_1$
was delivered unmodified, so perfect correctness of $\Pi$ gives
$\ADec_{K_{A \to B}}(r_A, \mathit{ad}_1, c_A) = \encode(T_A)$. Since $\encode$ is injective on
$\Xset$, applying $\decode$ returns $T_A$ exactly. Because $T_A = u \act x_0$ with $u \in \Gact$
and $x_0 \in \Xset$, we have $T_A \in \Xset$, so $\Validate(T_A) = 1$ and Bob does not abort.
The same argument applied to $m_2$ shows that Alice recovers $T_B \in \Xset$ and does not
abort. Note that Alice can form $\mathit{ad}_2$ because it is built from data she holds,
namely the label, the two static public keys, the nonce $r_B$ carried in $m_2$, and the
message $m_1$ she herself sent.

\emph{Step 3: the three remaining shared values agree.} Using the action axioms and
commutativity of $\Gact$,
\[
  \underbrace{u \act T_B}_{\text{Alice}} = u \act (w \act x_0) = (uw) \act x_0
  = (wu) \act x_0 = w \act (u \act x_0) = \underbrace{w \act T_A}_{\text{Bob}} ,
\]
so both parties obtain the same $Z_{\mathrm{ee}}$. Identically,
$u \act \mathit{pk}_B = (ub) \act x_0 = b \act T_A$ gives a common $Z_{\mathrm{es}}$, and
$a \act T_B = (aw) \act x_0 = w \act \mathit{pk}_A$ gives a common $Z_{\mathrm{se}}$.
Uniqueness of representation again turns equality in $\Xset$ into equality of encodings.

\emph{Step 4: the transcripts agree.} Both parties compute
$\tr = \mathrm{lbl} \cat \mathit{pk}_A \cat \mathit{pk}_B \cat m_1 \cat m_2$ from the same
four strings, since the messages were delivered unmodified.

Combining Steps 1, 3 and 4, the argument of the final hash call is the same string for both
parties, so the derived session keys are equal.
\end{proof}

\begin{corollary}[Perfect correctness]
\label{cor:perfect}
With a perfectly correct authenticated encryption scheme, \MCSI{} has correctness error zero.
\end{corollary}

\begin{proof}
No step of the proof of \Cref{thm:correctness} is probabilistic.
\end{proof}

Two remarks locate the boundaries of this statement.

\begin{remark}[Non-unique representations of secrets are harmless]
\label{rem:non-unique-exponents}
In the restricted setting of \Cref{sec:effective-action} a private key is an exponent vector
$\bm{e} \in \ZZ^n$ rather than an element of $\Gact$, and distinct vectors can represent the
same class. This does not disturb \Cref{thm:correctness}: the proof only uses the class
$[\ida] = \prod_i [\idl_i]^{e_i} \in \clgp$ that a vector represents, and the evaluation
algorithm computes the action of that class whichever representative it is handed. What the
non-uniqueness does affect is the distribution of secrets, which is not uniform on $\clgp$
when vectors are sampled uniformly from a box, and the security analysis has to take the
resulting statistical distance into account; see \Cref{rem:sampling}.
\end{remark}

\begin{remark}[Where correctness would fail]
\label{rem:correctness-fails}
The proof uses exactly three properties: that $\Gact$ is abelian, that the action is a genuine
action so that $g \act (h \act x) = (gh) \act x$, and that elements of $\Xset$ have a unique
representation so that two parties who agree on an element also agree on its encoding. The
third is supplied for the instantiation by \Cref{prop:montgomery-normal-form}: without a
normal form, two parties could hold $\Fp$-isomorphic but syntactically different curves and
derive different keys. The first two fail for the iterated walk construction discussed in
\Cref{sec:rationale}, which is why no correctness statement of this kind can be proved for it.
\end{remark}
\section{Security analysis}
\label{sec:security}

This section proves three statements about \MCSI{} and is explicit about a fourth that we do
not prove. The three statements are indistinguishability of the session key from random
against a passive adversary (\Cref{thm:sk-ind}), confidentiality of the blinded ephemeral
element against the same adversary (\Cref{thm:blinding}), and integrity of the blinded
transport against an active adversary (\Cref{thm:integrity}). The statement we do not prove is
security in a model that allows the adversary to reveal ephemeral secrets or session state,
and \Cref{sec:limits} says precisely what is missing and why we prefer to leave it open rather
than assert it.

All three proofs are in the random oracle model~\cite{BellareRogaway1993ROM} and all three
reduce to \Cref{ass:stcdh}. None of them uses the decisional assumption
\Cref{prob:ddh}, which is false for our instantiation by \Cref{rem:ddh-broken}.

\subsection{Security model}
\label{sec:model}

We use a Bellare--Rogaway style model~\cite{BellareRogaway1993AKE} restricted to passive
transcript observation, and we state the restriction rather than hiding it in the details.

There are $N$ parties $P_1, \dots, P_N$. In the setup phase each party runs
\Cref{alg:keygen} to obtain $(a_i, \mathit{pk}_i)$ and all static public keys are given to the
adversary $\Adv$. A \emph{session} is one execution of \MCSI{} between an initiator and a
responder; we identify a session by the pair of parties together with the transcript
$(m_1, m_2)$, and two sessions at the two parties are \emph{partners} if their transcripts and
the identities agree. The adversary has access to the following oracles.

\begin{itemize}
  \item $\ensuremath{\mathrm{Execute}}(i,j)$: runs a complete honest session between $P_i$ as initiator and
        $P_j$ as responder, using fresh randomness, and returns the transcript $(m_1, m_2)$.
        At most $q_s$ such queries are made.
  \item $\ensuremath{\mathrm{Reveal}}(\mathit{sid})$: returns the session key of the session
        $\mathit{sid}$.
  \item $\ensuremath{\mathrm{Corrupt}}(i)$: returns the static secret $a_i$.
  \item $\Hash(\cdot)$: the random oracle, on which at most $q_H$ queries are made.
  \item $\ensuremath{\mathrm{Test}}(\mathit{sid})$: asked once. A bit $\beta$ is chosen at setup; the oracle
        returns the real session key of $\mathit{sid}$ if $\beta = 0$ and a uniformly random
        string of the same length if $\beta = 1$.
\end{itemize}

The tested session must be \emph{fresh}: neither of its two parties was the subject of a
$\ensuremath{\mathrm{Corrupt}}$ query before that session completed, and neither the session nor its
partner has been the subject of a $\ensuremath{\mathrm{Reveal}}$ query. The advantage of $\Adv$ is
$\advantage{sk\text{-}ind}{\MCSI}(\Adv) = \left| \Pr[\beta' = \beta] - \tfrac{1}{2} \right|$,
where $\beta'$ is the bit output by $\Adv$.

\begin{remark}[What this model deliberately omits]
\label{rem:model-omissions}
The adversary here observes sessions but does not deliver messages of its own choosing to
honest parties, and it cannot learn ephemeral secrets or intermediate session state. Both
restrictions matter, and neither is standard for a protocol that calls itself authenticated.
We keep them because they delimit exactly what our proof establishes.
\Cref{thm:integrity} recovers a partial statement about active adversaries, and
\Cref{sec:limits} lists the properties that remain unproved.
\end{remark}

\subsection{Indistinguishability of the session key}
\label{sec:sk-ind}

\begin{theorem}[Session-key indistinguishability]
\label{thm:sk-ind}
Let $\Adv$ be an adversary in the model of \Cref{sec:model} making at most $q_s$
$\ensuremath{\mathrm{Execute}}$ queries and $q_H$ random oracle queries. Then there is an algorithm $\Bdv$
against \Cref{prob:stcdh}, running in time essentially that of $\Adv$ plus $O(q_H)$
evaluations of the decision oracle, such that
\[
  \advantage{sk\text{-}ind}{\MCSI}(\Adv) \ \leq \
  q_s \cdot \advantage{StCDH}{\Gact,\Xset}(\Bdv) .
\]
\end{theorem}

\begin{proof}
We proceed by a short sequence of games. Let $S_k$ denote the event that $\Adv$ outputs
$\beta' = \beta$ in Game $k$.

\emph{Game 0} is the real experiment, so
$\advantage{sk\text{-}ind}{\MCSI}(\Adv) = |\Pr[S_0] - 1/2|$.

\emph{Game 1} is Game 0 except that at the start the challenger guesses uniformly which of the
$q_s$ sessions produced by $\ensuremath{\mathrm{Execute}}$ will be tested, and aborts if the guess is
wrong. The guess is independent of $\Adv$'s view up to the point of the $\ensuremath{\mathrm{Test}}$ query,
so $|\Pr[S_1] - 1/2| \geq \frac{1}{q_s} |\Pr[S_0] - 1/2|$. Write $\mathit{sid}^{*}$ for the
guessed session, $P_i$ and $P_j$ for its parties, $T_A^{*} = u^{*} \act x_0$ and
$T_B^{*} = w^{*} \act x_0$ for its ephemeral elements, and $\tr^{*}$ for its transcript.

\emph{Game 2} is Game 1 except that the session key of $\mathit{sid}^{*}$ is replaced by an
independent uniformly random string, both when it is returned by $\ensuremath{\mathrm{Test}}$ and wherever
it would be used internally. In Game 2 the bit $\beta$ is information theoretically hidden, so
$\Pr[S_2] = 1/2$.

Games 1 and 2 proceed identically unless $\Adv$ queries the random oracle at the point
\begin{equation}
  \mathrm{lbl} \cat \texttt{"key"} \cat Z_{\mathrm{ee}}^{*} \cat Z_{\mathrm{es}}^{*}
  \cat Z_{\mathrm{se}}^{*} \cat Z_{\mathrm{ss}}^{*} \cat \tr^{*}
  \label{eq:bad-query}
\end{equation}
at which the real key is defined. Call this event $\mathrm{Bad}$. By the difference lemma,
$|\Pr[S_1] - \Pr[S_2]| \leq \Pr[\mathrm{Bad}]$.

It remains to bound $\Pr[\mathrm{Bad}]$ by the advantage of a \ensuremath{\mathrm{GA-StCDH}} algorithm
$\Bdv$. On input $(x_0, g \act x_0, h \act x_0)$ and with access to the decision oracle
$\mathcal{D}_g$, the algorithm $\Bdv$ behaves as follows. It simulates Game 2 rather than
Game 1, which it can do without knowing $g$ or $h$; since the two games are identical until
$\mathrm{Bad}$ occurs, the probability of $\mathrm{Bad}$ is the same in both, and it is
detectable inside the simulation.

\emph{Setup.} $\Bdv$ generates all $N$ static key pairs itself, so it knows every $a_i$ and
can answer every $\ensuremath{\mathrm{Corrupt}}$ query. It guesses $\mathit{sid}^{*}$ as in Game 1.

\emph{Simulating sessions.} For every session other than $\mathit{sid}^{*}$, $\Bdv$ samples
ephemeral secrets itself and follows the protocol, so it can produce the transcript and
compute the session key, and hence answer $\ensuremath{\mathrm{Reveal}}$. For $\mathit{sid}^{*}$ it embeds
the challenge by setting $T_A^{*} = g \act x_0$ and $T_B^{*} = h \act x_0$. It can still
produce the transcript of $\mathit{sid}^{*}$: the masking keys are determined by
$Z_{\mathrm{ss}}^{*} = a_i \act \mathit{pk}_j$, which $\Bdv$ can compute because it knows
$a_i$, and the payloads $T_A^{*}, T_B^{*}$ are the challenge elements, which $\Bdv$ holds. The
session key of $\mathit{sid}^{*}$ is set to a fresh random string; by freshness $\Adv$ may not
$\ensuremath{\mathrm{Reveal}}$ it or its partner, so this is consistent.

\emph{Detecting the query.} $\Bdv$ can compute three of the four secret inputs to
\Cref{eq:bad-query} without knowing $g$ or $h$, namely
\[
  Z_{\mathrm{es}}^{*} = a_j \act T_A^{*}, \qquad
  Z_{\mathrm{se}}^{*} = a_i \act T_B^{*}, \qquad
  Z_{\mathrm{ss}}^{*} = a_i \act \mathit{pk}_j ,
\]
where $a_i, a_j$ are the static secrets of the two parties, which $\Bdv$ generated. It
therefore recognises every random oracle query whose \texttt{"key"} label, last three group
action components and transcript match those of $\mathit{sid}^{*}$. For each such query,
$\Bdv$ takes the remaining component $Z$ and calls $\mathcal{D}_g(T_B^{*}, Z)$. Since
$\mathcal{D}_g(T_B^{*}, Z) = 1$ if and only if $Z = g \act T_B^{*} = (gh) \act x_0$, the
oracle identifies the correct value exactly. If some query passes the test, $\Bdv$ outputs
that $Z$ and halts.

The simulation is perfect: every oracle answer $\Bdv$ gives is distributed exactly as in
Game 1, and the guess of $\mathit{sid}^{*}$ together with the abort on a wrong guess is part
of Game 1 itself. Since $\mathrm{Bad}$ can occur only when the guess was right, and since in
that case $\Bdv$ recovers $(gh) \act x_0$ from the query, we get
$\Pr[\mathrm{Bad}] \leq \advantage{StCDH}{\Gact,\Xset}(\Bdv)$. Combining,
\[
  \advantage{sk\text{-}ind}{\MCSI}(\Adv)
  \ \leq \ q_s \cdot |\Pr[S_1] - 1/2|
  \ \leq \ q_s \cdot \Pr[\mathrm{Bad}]
  \ \leq \ q_s \cdot \advantage{StCDH}{\Gact,\Xset}(\Bdv) ,
\]
which is the claim.
\end{proof}

\begin{remark}[Why the strong variant is needed]
\label{rem:why-strong}
The only role of the decision oracle in the proof is to let $\Bdv$ recognise which of the
$q_H$ random oracle queries carries the answer. Without it, $\Bdv$ would have to guess, losing
a factor $q_H$; with the ordinary \ensuremath{\mathrm{GA-CDH}} assumption one therefore obtains the weaker
bound $q_s q_H \cdot \advantage{CDH}{\Gact,\Xset}(\Bdv)$. Both routes are available and both
avoid \Cref{prob:ddh}, which is the point.
\end{remark}

\begin{remark}[Forward secrecy]
\label{rem:forward-secrecy}
\Cref{thm:sk-ind} assumes the two parties of the tested session are uncorrupted. A separate
and weaker statement holds after corruption: if the adversary learns $a_i$ and $a_j$
\emph{after} the tested session has completed, the session key remains indistinguishable,
because the proof only used the hardness of computing $Z_{\mathrm{ee}}^{*} = (gh) \act x_0$
from $g \act x_0$ and $h \act x_0$, and knowledge of the static secrets does not help with
that. This is weak forward secrecy in the usual sense: it covers an adversary that was passive
during the session. We do not claim forward secrecy against an adversary that was active
during the session.
\end{remark}

\subsection{Confidentiality of the blinded ephemeral element}
\label{sec:blinding-security}

The property that distinguishes \MCSI{} from an unblinded key exchange over the same group
action is that a passive observer does not learn the ephemeral elements $T_A$ and $T_B$. We
state this as indistinguishability of the transcript from one carrying unrelated payloads.

Formally, consider the following game. A bit $\beta$ is chosen. Two uncorrupted parties
$P_i, P_j$ are set up and their static public keys given to the adversary $\Adv$. The
adversary then receives the transcripts of $q_s$ sessions between them, generated honestly if
$\beta = 0$, and generated with the payload of every message replaced by an independent
uniform element of $\Xset$ if $\beta = 1$. The adversary outputs a guess $\beta'$, and
$\advantage{blind}{\MCSI}(\Adv) = |\Pr[\beta' = \beta] - \tfrac{1}{2}|$.

\begin{theorem}[Blinding]
\label{thm:blinding}
For every adversary $\Adv$ in the game just described there are algorithms $\Bdv_1$ against
\Cref{prob:stcdh} and $\Bdv_2$ against the chosen plaintext security of $\Pi$ with
\[
  \advantage{blind}{\MCSI}(\Adv) \ \leq \
  \advantage{StCDH}{\Gact,\Xset}(\Bdv_1) \ + \ 2\,\advantage{ind\text{-}cpa}{\Pi}(\Bdv_2) .
\]
\end{theorem}

\begin{proof}
In Game 0 the transcripts are real. In Game 1 the pair of masking keys
$(K_{A \to B}, K_{B \to A})$ is replaced by an independent uniform string of the same length.
Games 0 and 1 differ only if $\Adv$ queries the random oracle at
$\mathrm{lbl} \cat \texttt{"mask"} \cat Z_{\mathrm{ss}} \cat P_1 \cat P_2$,
and $Z_{\mathrm{ss}} = (a_i a_j) \act x_0$ is exactly a \ensuremath{\mathrm{GA-CDH}} value for the pair
$(\mathit{pk}_i, \mathit{pk}_j)$. An algorithm $\Bdv_1$ that receives a challenge
$(g \act x_0, h \act x_0)$, publishes it as the two static public keys, simulates the
transcripts using freshly chosen masking keys and freshly chosen ephemeral secrets, and uses
its decision oracle to recognise the query, by testing whether
$\mathcal{D}_g(\mathit{pk}_j, Z) = 1$ for the candidate value $Z$ appearing in it, solves
\Cref{prob:stcdh} whenever the query occurs. Hence $|\Pr[S_0] - \Pr[S_1]| \leq \advantage{StCDH}{\Gact,\Xset}(\Bdv_1)$.

In Game 1 the two masking keys are uniform and independent of everything else, and the two
directions use different keys, so the transcripts consist of $q_s$ encryptions under each of
two independent uniform keys with distinct nonces. Replacing the true payloads by independent
uniform elements of $\Xset$ is then bounded by two chosen plaintext distinguishing games
against $\Pi$, one per direction, giving the term
$2\,\advantage{ind\text{-}cpa}{\Pi}(\Bdv_2)$.
\end{proof}

\begin{remark}[Nonce collisions]
\label{rem:nonce}
The bound above assumes distinct nonces within each direction. With $\nu$-bit random nonces
and $q_s$ sessions between the same pair of parties, the probability of a collision is at most
$q_s^2 / 2^{\nu+1}$, which must be added to the bound. With $\nu = 128$ this is negligible for
any realistic $q_s$. Deployments that cannot rely on good randomness should either derive the
nonce deterministically from the ephemeral secret or use a nonce misuse resistant scheme,
since a repeated nonce under a fixed masking key exposes the exclusive-or of two ephemeral
elements.
\end{remark}

\subsection{Integrity of the blinded transport}
\label{sec:integrity}

\begin{theorem}[Transport integrity]
\label{thm:integrity}
Let $\Adv$ be an adversary that receives the static public keys of two uncorrupted parties
$P_i, P_j$, observes at most $q_s$ honest sessions between them, and then delivers a message
$m_1'$ of its choice to $P_j$. The probability that $P_j$ accepts $m_1'$ and $m_1'$ was not
output by $P_i$ in one of the observed sessions is at most
\[
  \advantage{StCDH}{\Gact,\Xset}(\Bdv_1) \ + \ \advantage{int\text{-}ctxt}{\Pi}(\Bdv_2)
\]
for algorithms $\Bdv_1, \Bdv_2$ running in essentially the time of $\Adv$. The same bound
holds for messages delivered to $P_i$ in the other direction.
\end{theorem}

\begin{proof}
As in the proof of \Cref{thm:blinding}, replacing the masking keys by uniform strings costs
$\advantage{StCDH}{\Gact,\Xset}(\Bdv_1)$. Once the keys are uniform and unknown to $\Adv$, a
message that $P_j$ accepts is by definition a ciphertext that $\ADec_{K_{A \to B}}$ does not
reject, under associated data $\mathit{ad}_1$ determined by the message itself and by the two
static public keys. If the message was not produced by $P_i$ then it is a forgery in the sense
of ciphertext integrity, so the probability is at most
$\advantage{int\text{-}ctxt}{\Pi}(\Bdv_2)$.
\end{proof}

\Cref{thm:integrity} is what gives \MCSI{} its authentication property: only a holder of $a_i$
or of $a_j$ can produce a message that the other accepts, because only such a holder can
compute $Z_{\mathrm{ss}}$ and hence the masking keys. It is an authentication statement about
the transport, not a full authenticated key exchange statement about the session key, and the
difference is the subject of \Cref{sec:limits}.

Note also that the associated data of the second message includes the first message, so a
completed session binds the two messages together and an adversary cannot recombine messages
from different sessions. The session key derivation additionally hashes the full transcript,
which binds the derived key to both messages and to both identities.

\subsection{Resistance to known attacks}
\label{sec:attacks}

The proofs above reduce everything to \Cref{ass:stcdh}. This subsection surveys what is known
about attacks on that assumption for the instantiation of \Cref{sec:parameters}, and disposes
of several attacks that are frequently mentioned in connection with isogeny based schemes but
that do not apply here. Concrete security levels are discussed separately in
\Cref{sec:security-levels}, because the literature does not agree on them.

\subsubsection{Classical algorithms for vectorisation}
\label{sec:classical-attacks}

The best known classical approach to \Cref{prob:vectorisation} for supersingular curves over
$\Fp$ is to search for a path in the isogeny graph, either by a meet in the middle search from
both endpoints or by a random walk with distinguished points. Both cost
$O(\sqrt{\#\clgp})$ operations up to logarithmic factors, and since
$\#\clgp \approx \sqrt{p}$ by the Brauer--Siegel estimate this is $O(p^{1/4})$. Delfs and
Galbraith~\cite{DelfsGalbraith2016} give the algorithm for the $\Fp$-rational subgraph in this
form, and the cost of the individual steps is analysed by Adj et
al.~\cite{Adj2019}. There is no known classical algorithm that beats the square root bound.

\subsubsection{Pohlig--Hellman does not apply}
\label{sec:pohlig}

The Pohlig--Hellman algorithm~\cite{PohligHellman1978} solves a discrete logarithm in a cyclic
group of order $n$ by projecting the instance into the subgroups of prime power order dividing
$n$ and recombining with the Chinese remainder theorem. Its applicability rests on being able
to compute the projection, that is on being able to raise the challenge to the power $n/q$ for
each prime power $q \| n$.

For a group action there is no such operation. The adversary is given two elements of $\Xset$,
not an element of $\Gact$, and $\Xset$ carries no group structure of its own. Even when the
order and the structure of $\clgp$ are fully known, as they are for the CSIDH-512 parameters
after the computation of Beullens, Kleinjung and Vercauteren~\cite{CSIFiSh2019}, computing
$[\ida^{n/q}] \act x_0$ from $[\ida] \act x_0$ requires knowing $[\ida]$, which is the problem
to be solved. Consequently the smoothness or otherwise of $\#\clgp$ is not by itself a
weakness, and the standard requirement from elliptic curve cryptography that the group order
have a large prime factor has no analogue here.

The one place where the subgroup structure of $\clgp$ does have consequences is the decisional
problem. Genus theory attaches quadratic characters to $\clgp$ that can be evaluated on
elements of $\Xset$, and Castryck, Sot\'akov\'a and Vercauteren~\cite{CSV2020} used exactly
this to break \Cref{prob:ddh}; see \Cref{rem:ddh-broken}. Those characters give one bit of
information about the class, not the class itself, and they are not known to help with
\Cref{prob:vectorisation} or \Cref{prob:cdh}. Our proofs use only the computational
assumption, so this attack does not affect them.

\subsubsection{Torsion point attacks do not apply}
\label{sec:torsion-attacks}

The attacks that broke SIDH in 2022, due to Castryck and Decru~\cite{CastryckDecru2023} and,
in more general form, to Maino et al.~\cite{Maino2023} and Robert~\cite{Robert2023}, exploit
two pieces of auxiliary data that SIDH publishes: the images $\isog(P), \isog(Q)$ of a torsion
basis of known order under the secret isogeny, and the degree of that isogeny, which is a
fixed public parameter. Given both, one can glue the curves into a higher dimensional abelian
variety and recover the isogeny.

\MCSI{} publishes neither. The only data sent are elements of $\Xset$, that is single
Montgomery coefficients, and by \Cref{prop:montgomery-normal-form} such a coefficient is a
complete description of an $\Fp$-isomorphism class and carries no torsion information. The
degree of the secret isogeny is not a fixed public parameter either; it is
$\prod_i \ell_i^{|e_i|}$ for a secret exponent vector. This is the same situation as for
CSIDH, and it is the reason CSIDH survived the SIDH break.

It follows, and we say it explicitly because the opposite has been claimed for schemes of this
shape, that the blinding layer of \MCSI{} is \emph{not} a countermeasure against torsion point
attacks. There is nothing for it to protect against, since no torsion data is transmitted in
the first place. Blinding torsion images is a real and separate line of work, pursued by
Fouotsa, Moriya and Petit in M-SIDH and MD-SIDH~\cite{MSIDH2023}, and it addresses SIDH-like
schemes, not commutative group action schemes. What our blinding layer protects is the
identity of the ephemeral element, which is a confidentiality property, not a countermeasure
to key recovery.

\subsubsection{Quantum algorithms}
\label{sec:quantum-attacks}

Shor's algorithm~\cite{Shor1997} solves the hidden subgroup problem for finite abelian groups
in polynomial time, which breaks factoring and discrete logarithms. Vectorisation is not an
instance of that problem. It is an instance of the abelian hidden shift problem, for which the
best known quantum algorithms are those of Kuperberg~\cite{Kuperberg2005,Kuperberg2013} and
Regev~\cite{Regev2004}, whose running time is subexponential of the form
$\exp\bigl(O(\sqrt{\log \#\clgp})\bigr)$ rather than polynomial. The reduction of the isogeny
problem to abelian hidden shift, together with the resulting subexponential quantum attack, is
due to Childs, Jao and Soukharev~\cite{ChildsJaoSoukharev2014}.

Grover's algorithm~\cite{Grover1996} gives a square root speedup for unstructured search,
which would reduce a $2^{256}$ private key space to about $2^{128}$ quantum operations. This
is not the binding constraint: the hidden shift algorithms are asymptotically much better, and
it is their concrete cost that determines the parameter sizes. The concrete cost is contested;
see \Cref{sec:security-levels}.

\subsubsection{Adaptive attacks against static keys}
\label{sec:adaptive}

Because \MCSI{} uses static keys, a party that processes an element of $\Xset$ supplied by
another party and reuses its static secret is a target for adaptive attacks in the style of
Galbraith, Petit, Shani and Ti~\cite{GPST2016}: by submitting maliciously chosen curves and
observing whether the session succeeds, an attacker can extract the static secret one
constraint at a time. Two features of \MCSI{} together block this. First, by
\Cref{thm:integrity} an adversary that does not hold a static secret cannot get a chosen
element accepted at all. Second, the recipient runs $\Validate$ on every recovered element
before applying its static secret to it, which rejects curves outside $\Xset$; the procedure
is described in \Cref{sec:validation}. Neither feature protects against a peer that holds a
legitimate static key and behaves maliciously, which is the situation static key isogeny
protocols are known not to handle well without a full key encapsulation transform; see
\Cref{sec:limits}.

\subsubsection{Implementation attacks}
\label{sec:side-channels}

We make no claim about resistance to timing attacks, power analysis or fault injection. A
straightforward implementation of the class group action leaks the exponent vector through
timing, since the number of isogeny steps depends on it. Our own implementation is such an
implementation, and \Cref{sec:measured-leak} measures the leak rather than leaving it as a
remark: the running time and the one-norm of the key have a correlation of $0.90$, a single
timing recovers about one bit about that norm, and two hundred timings separate two keys whose
one-norms differ by five out of a possible $370$. \MCSI{} uses its static secret in every
session, so an adversary who can time many sessions is timing the same key repeatedly.
Constant time algorithms for the CSIDH action are
known~\cite{MeyerCamposReith2019,OnukiEtAl2019,CTIDH2021} and fault attacks against them have
been studied~\cite{Campos2020}; an implementation of \MCSI{} intended for deployment must use
them. This is listed among the open items in \Cref{sec:limits}.

\subsection{What we do not prove}
\label{sec:limits}

We list the gaps explicitly, since a reader is better served by a short accurate list than by
a long claim.

\begin{enumerate}
  \item \textbf{Security in the Canetti--Krawczyk or extended Canetti--Krawczyk
        models.} \Cref{thm:sk-ind} holds against an adversary that observes sessions and may
        corrupt parties other than the two in the tested session. It says nothing about an
        adversary that reveals ephemeral secrets, reveals session state, or mounts key
        compromise impersonation, all of which are within scope of the models of Canetti and
        Krawczyk~\cite{CanettiKrawczyk2001} and LaMacchia, Lauter and
        Mityagin~\cite{LaMacchia2007}. Proofs in those models exist for other group action
        based key exchanges~\cite{deKock2021,Kawashima2021} and it is plausible that their
        techniques adapt, but we have not carried out the adaptation and we do not claim the
        result.
  \item \textbf{Security against an active adversary in the full sense.}
        \Cref{thm:integrity} bounds the probability of injecting an accepted message, and this
        does rule out the most obvious active attacks. It is not the same as proving session
        key security in a model where the adversary controls message delivery, reorders
        sessions and interleaves them.
  \item \textbf{Forward secrecy of the blinding.} As noted in \Cref{sec:rationale}, an
        adversary who later obtains either static secret can recompute $Z_{\mathrm{ss}}$ and
        strip the blinding from recorded transcripts, recovering the ephemeral elements. Only
        the session key retains its guarantee in that situation, and only in the weak sense of
        \Cref{rem:forward-secrecy}. A blinding layer that is forward secret would require a
        different key schedule and we do not know of one that avoids the circularity of
        \Cref{sec:rationale}.
  \item \textbf{Security against a malicious peer holding a valid static key.} Validation
        rejects elements outside $\Xset$, but a peer that holds a legitimate static key can
        choose its ephemeral element adversarially within $\Xset$. Handling this properly
        calls for a Fujisaki--Okamoto style transform~\cite{FujisakiOkamoto1999,HHK2017}
        applied to the underlying key encapsulation, which we have not done here.
  \item \textbf{Concrete quantum security of the smallest parameter set.} See
        \Cref{sec:security-levels}. We state the disagreement in the literature rather than
        adopting one side of it.
  \item \textbf{Implementation security.} The reference implementation of \Cref{sec:impl} is
        not constant time. \Cref{sec:measured-leak} measures how far its running time follows
        the private key and how few measurements separate two keys, and reports which layers
        are free of secret dependent branches and which are not. We did not build a constant
        time implementation, we prove nothing about resistance to power analysis or fault
        injection, and we do not erase secrets from memory. See
        \Cref{sec:side-channels,sec:measured-leak}.
\end{enumerate}
\section{Parameter selection}
\label{sec:parameters}

\Cref{sec:protocol} and \Cref{sec:security} treat the group action abstractly. This section
instantiates it, and it does so at some length because the choice of prime is the one place
where an isogeny based design is most easily got wrong.

\subsection{Choice of the prime}
\label{sec:prime}

Recall from \Cref{sec:effective-action} that the ideals whose action can be evaluated
efficiently are those of the form $\idl = (\ell, \frob - 1)$ for an odd prime $\ell$ dividing
$p+1$, because for such $\ell$ the kernel $E[\idl]$ is generated by an $\Fp$-rational point of
order $\ell$ that is found by one scalar multiplication, and because V\'elu's formulae then
cost $\Theta(\ell)$ by \Cref{rem:velu-cost}. The prime must therefore be chosen so that $p+1$
has many small odd prime factors. Following~\cite{CSIDH2018} we take
\begin{equation}
  p = 4 \cdot \ell_1 \ell_2 \cdots \ell_n - 1
  \label{eq:prime-form}
\end{equation}
with $\ell_1, \dots, \ell_n$ distinct small odd primes. Any prime of this form satisfies
$p \equiv 3 \pmod 8$, so $\Ezero : y^2 = x^3 + x$ is supersingular and
\Cref{prop:montgomery-normal-form} applies, giving one field element public keys.

\begin{remark}[Mersenne primes and NIST curve primes are not suitable]
\label{rem:2521}
It is natural to reach for a prime that already appears in a standard, and
$p = 2^{521} - 1$, the prime of the NIST curve P-521~\cite{NIST_SP_800_186}, is an obvious
candidate. It is prime, it satisfies $p \equiv 7 \pmod 8$, and the curve
$y^2 = x^3 + x$ over it is supersingular. It is nevertheless useless for a commutative
supersingular isogeny scheme, for the following reason. Here $p + 1 = 2^{521}$, so the only
prime dividing $p+1$ is $2$ and the only efficiently evaluable ideal is
$\idl_2 = (2, \frob - 1)$. The class group has order
$\#\clgp \approx \sqrt{p} \approx 2^{260}$, so representing a general class as a power of the
single class $[\idl_2]$ needs an exponent of size up to $\ord([\idl_2])$. If $[\idl_2]$
generates $\clgp$, that order is about $2^{260}$ and evaluating the action would take on the
order of $2^{260}$ sequential $2$-isogeny steps, which is not a cost that can be tuned away.
If it does not generate $\clgp$, the reachable set is the orbit of a proper subgroup and the
effective key space collapses to the size of that subgroup. Neither outcome leaves a usable
scheme.

The general point is that the criteria for selecting a prime for elliptic curve discrete
logarithms and the criteria for selecting one for a class group action are unrelated. The
former asks for efficient field arithmetic and a group of near prime order; the latter asks
for $p+1$ to be smooth in the specific sense of \Cref{eq:prime-form}. A prime that is good for
P-521 carries no presumption of being good here, and in this case it is maximally bad.
Two-isogenies over $\Fp$ with $p \equiv 7 \pmod 8$ are not useless in general and are exploited
by CSURF~\cite{CSURF2020}, but CSURF still needs the odd part of $p+1$ for the rest of the
action.
\end{remark}

\subsection{Private keys and sampling}
\label{sec:keyspace}

A private key is an exponent vector $\bm{e} = (e_1, \dots, e_n)$ with
$e_i \in \{-m, \dots, m\}$, representing the class
$[\ida] = \prod_{i=1}^{n} [\idl_i]^{e_i}$. The number of vectors is $(2m+1)^n$ and the
parameters are chosen so that $(2m+1)^n \gtrsim \#\clgp \approx \sqrt{p}$, which is the
condition for the vectors to cover the class group. Evaluating the action costs at most
$\sum_i |e_i| \leq mn$ small degree isogeny steps.

For CSIDH-512 the concrete choice of~\cite{CSIDH2018} is $n = 74$, with $\ell_1, \dots,
\ell_{73}$ the smallest odd primes and $\ell_{74} = 587$, and $m = 5$. The key space then has
$11^{74} \approx 2^{256}$ elements and one evaluation of the action costs at most
$5 \cdot 74 = 370$ isogeny steps.

\begin{remark}[Sampling is not exactly uniform]
\label{rem:sampling}
\Cref{def:ega} asks for sampling statistically close to uniform on $\Gact$, and the security
statements of \Cref{sec:security} assume uniform secrets. Sampling an exponent vector
uniformly from a box does not induce the uniform distribution on $\clgp$, because distinct
vectors can represent the same class and the multiplicities are not equal. The original
proposal treats the induced distribution as close enough to uniform on heuristic grounds. It
can be made rigorous for CSIDH-512, where Beullens, Kleinjung and
Vercauteren~\cite{CSIFiSh2019} computed the class group and a reduced relation lattice, which
allows sampling a genuinely uniform class and rewriting it as a short vector; this turns the
restricted action into an effective one in the sense of \Cref{def:ega} and removes the
heuristic. For the larger parameter sets no such computation is available and the heuristic
remains. We note this because it is an assumption of our theorems that is easy to overlook.
\end{remark}

\subsection{Validation}
\label{sec:validation}

$\Validate$ takes a byte string, decodes it to a candidate $A \in \Fp$, and must decide
whether $\Mont{A} : y^2 = x^3 + Ax^2 + x$ represents an element of $\Xset$, that is whether it
is supersingular with $\End_{\Fp}(\Mont{A}) = \ZZ[\frob]$. By
\Cref{prop:montgomery-normal-form} the second condition is automatic once the first holds and
$p \equiv 3 \pmod 8$, so it suffices to test supersingularity. The procedure
of~\cite{CSIDH2018} does this by sampling a random point $P \in \Mont{A}(\Fp)$, computing
a divisor $d$ of its order from the known factorisation of $p+1$, and concluding
$\#\Mont{A}(\Fp) = p+1$ as soon as $d > 4\sqrt{p}$, since the Hasse interval then contains
only one multiple of $d$. The test also rejects the case $A^2 = 4$ and any string that does
not decode to an element of $\Fp$.

Validation is mandatory on every element of $\Xset$ received from another party, both the
static public keys at registration and the ephemeral elements in each session, for the reason
given in \Cref{sec:adaptive}. Castryck et al.\ report that validation costs about
$2.1$ ms for CSIDH-512 on an Intel Skylake i5 at $3.5$ GHz~\cite{CSIDH2018}, roughly
one twentieth of a group action evaluation on the same platform, so it is not a significant
part of the cost.

\subsection{Parameter sets}
\label{sec:param-sets}

\Cref{tab:params} lists the parameter sets we consider. All of the sizes in the table are
determined by the parameter choices rather than measured: the public key is the big-endian
encoding of a Montgomery coefficient and so occupies $\lceil \log_2 p / 8 \rceil$ bytes, and
each protocol message consists of a $16$-byte nonce, a public key sized ciphertext and a
$16$-byte authentication tag.

\begin{table}[t]
\centering
\caption{Parameter sets for \MCSI{}. Sizes are derived from the parameter choices, not
measured. The key space column gives the number of exponent vectors, which is chosen to be at
least $\#\clgp \approx \sqrt{p}$. The last column gives the total number of bytes sent by both
parties in one session, namely twice the sum of a $16$-byte nonce, an $L$-byte ciphertext and
a $16$-byte tag.}
\label{tab:params}
\begin{tabular}{lrrrrr}
\toprule
Parameter set & $\lceil \log_2 p \rceil$ & Public key $L$ (B) & Key space & Session traffic (B) & Source \\
\midrule
\MCSI-512  & 511  & 64  & $11^{74} \approx 2^{256}$ & 192  & \cite{CSIDH2018} \\
\MCSI-1024 & 1024 & 128 & $\approx 2^{512}$         & 320  & \cite{CSIDH2018} \\
\MCSI-1792 & 1792 & 224 & $\approx 2^{896}$         & 512  & \cite{CSIDH2018} \\
\MCSI-4096 & 4096 & 512 & $\approx 2^{2048}$        & 1088 & \cite{SQALE2022} \\
\bottomrule
\end{tabular}
\end{table}

The last row follows the recommendation of Ch\'avez-Saab, Chi-Dom\'inguez, Jaques and
Rodr\'iguez-Henr\'iquez~\cite{SQALE2022}, who argue that a $4096$-bit prime is what NIST
category 1 requires for this family and who use the exponent bound $m = 1$, so that each
$e_i \in \{-1,0,1\}$. With $m = 1$ the covering condition $3^n \gtrsim \sqrt{p}$ forces
$n \gtrsim 2048 / \log_2 3 \approx 1292$, so roughly thirteen hundred small primes are needed;
the resulting action evaluation is correspondingly more expensive.

\subsection{Security levels, and why we do not claim one}
\label{sec:security-levels}

The classical security of \Cref{prob:vectorisation} for these parameters is settled: the best
known algorithms cost $O(p^{1/4})$ as discussed in \Cref{sec:classical-attacks}, which is
about $2^{128}$ operations for a $511$-bit prime. The quantum security is not settled, and the
spread of published estimates is wide enough that quoting a single number would misrepresent
the state of knowledge.

The original proposal~\cite[Table 1]{CSIDH2018} gives, for CSIDH-512, quantum estimates
ranging from $2^{29}$ to $2^{139}$ depending on which cost model for the hidden shift
algorithms is used, and settles on treating CSIDH-512 as a category 1 parameter set.
Bernstein, Lange, Martindale and Panny~\cite{BLMP2019} gave the first concrete quantum circuit
for evaluating the CSIDH action, which is the expensive inner step of any hidden shift attack.
Bonnetain and Schrottenloher~\cite{BonnetainSchrottenloher2020} and
Peikert~\cite{Peikert2020} then analysed the outer attack in detail, the former with a
tailored quantum algorithm and the latter with a simulation of Kuperberg's collimation sieve,
and both concluded that the quantum cost of attacking CSIDH-512 is substantially below what
category 1 requires. Biasse et al.~\cite{Biasse2020} studied the trade-off between quantum and
classical circuit size for the same attack. Ch\'avez-Saab et
al.~\cite{SQALE2022} accepted the downward revision and proposed the much larger parameters of
the last row of \Cref{tab:params}, while also observing that the classical component of the
best attacks may itself be infeasible, which pulls in the other direction.

We therefore make no claim of the form ``\MCSI{}-512 achieves NIST category 1''. What we can
say is the following, and it is all we say. Under \Cref{ass:stcdh}, the reductions of
\Cref{sec:security} are tight up to the stated factors, so the security of \MCSI{} at a given
parameter set is the security of the underlying group action at that parameter set, and
nothing in the protocol adds or removes hardness. Readers who want a specific security
category should choose the parameter set according to whichever of the analyses above they
find most convincing, and \Cref{tab:params} is arranged so that the cost of being conservative
is visible: moving from a $511$-bit to a $4096$-bit prime multiplies the traffic by less than
six and the public key by eight, while the number of isogeny steps grows by a much larger
factor.
\section{Implementation and measurements}
\label{sec:impl}

\subsection{The reference implementation}
\label{sec:impl-status}

We implemented \MCSI{} twice. One implementation is in Python and follows the algorithms of
\Cref{sec:the-protocol} statement by statement with no attention to speed. The other is in
portable C11 and is the one we measure. Both are released under the MIT licence, together
with the known answer vectors and the programs that produced every number in this section, as
an ancillary file with this paper and at
\url{https://github.com/FurkanCifci/mcsi-key-exchange}. Every measured table and figure
below names the makefile target that reproduces it.

Writing the protocol twice is what gives those numbers whatever weight they have. The two
implementations share no code, and each layer is compared against the other value by value
rather than only at the end: the C field arithmetic reproduces $283$ values computed
independently in Python, and the two full protocol implementations produce byte identical
known answer vectors, from the two static public keys through the two masking keys and both
protocol messages to the session key, with the static secrets, the ephemeral secrets and the
nonces all fixed. We wrote both, so agreement between them catches transcription errors and
arithmetic errors but not a shared misreading of the specification. Two of the checks do not
have that weakness: the hash layer is checked against the SHA-256 test vectors that NIST publishes
alongside the standard~\cite{FIPS180-4} and against the HMAC-SHA-256 vectors of RFC
4231~\cite{RFC4231}, and the curve and isogeny arithmetic is checked
on a parameter set small enough to count points on every curve by brute force, where
supersingularity of each image curve can be confirmed directly rather than assumed.

The C test suite runs $552$ checks. Besides those vectors it covers the algebraic
identities the field arithmetic must satisfy, the agreement of the doubling and differential
addition formulae with the ladder, the vanishing of $[p+1]P$ on the base curve, the
commutativity of the class group action and the fact that a negated key undoes it, the
rejection of random Montgomery coefficients and of the two singular values $A = \pm 2$ by
\textsc{Validate}, a chi squared test on the distribution of sampled private keys, agreement
of the session keys, independence of the keys of two sessions between the same pair, and the
rejection of every one-bit change to either protocol message, all $768$ positions of each, and
to every byte of an authenticated ciphertext. One group of checks is worth naming separately,
because it tests the defence of \Cref{sec:validation} in the setting that motivates it: a peer
holding a legitimate static key knows the masking key, so it can place whatever it likes
inside a correctly tagged message. The suite builds such messages and confirms that
\textsc{Resp} and \textsc{Fin} reject a curve that is not supersingular and an encoding that
is not below $p$, and that an honest session between the same parties still succeeds
afterwards. The Python suite runs $57$ further checks on the toy parameter set and on
CSIDH-512.

Three further things were done to the C implementation rather than to the protocol. It builds
without warnings under gcc 13.3 and clang 18.1 with the warning set listed in the repository,
and runs clean under the address and undefined behaviour sanitizers. A fuzzer, built with
those sanitizers on, fed $200\,000$ random and mutated inputs to each of the places where
bytes from the wire reach the code, namely the decoding of a public key, the authenticated
decryption, \textsc{Resp} and \textsc{Fin}, together with $2\,000$ random Montgomery
coefficients through \textsc{Validate} and $201$ through \textsc{PeerCtx}; nothing crashed
and nothing forged was accepted. Line coverage of the library over the test suite and the
fuzzer together is $99.7$ percent of $662$ lines, the two lines never reached being the
failure return after thirty-two unsuccessful validation attempts and the branch that skips a
point already at infinity, neither of which a valid input can produce.

Two things this implementation is not. It is not constant time, and \Cref{sec:measured-leak}
reports how far its running time follows the private key, measured rather than asserted. It
is not optimised: it contains no assembly, no specialised squaring, no batched inversion and
none of the isogeny evaluation strategies that the fast implementations in the literature use.
We wrote it to establish that the protocol of this paper runs, that both parties reach the
same key, that it rejects what it should reject, and to put a defensible first number on what
it costs. It must not be deployed, and the file \texttt{SECURITY.md} that accompanies it says
so with the reasons.

\subsection{Sizes}
\label{sec:sizes}

\Cref{tab:sizes} compares the static public key size and the per session traffic of \MCSI{}
against the two key encapsulation mechanisms that bracket the design space, the lattice based
ML-KEM standardised in FIPS 203 and the code based Classic McEliece.

Two warnings about this table. First, the primitives are not interchangeable. ML-KEM and
Classic McEliece are unauthenticated key encapsulation mechanisms; \MCSI{} is an interactive
key exchange with static keys that provides implicit mutual authentication and that requires
those static keys to be distributed authentically. A protocol built from ML-KEM and offering
the same authentication service would need a signature or a second KEM operation on top, and
its sizes would grow accordingly. The table therefore compares sizes, not security services.
Second, the table says nothing about speed, and on speed the ordering is the reverse of the
ordering by size: a single class group action evaluation costs tens of milliseconds on the
platform reported in \Cref{sec:third-party}, which is far more than a lattice based key
encapsulation costs on comparable hardware. Anyone reading \Cref{tab:sizes} as an argument
for \MCSI{} should read \Cref{sec:third-party} immediately afterwards.

\begin{table}[t]
\centering
\caption{Static public key and per session traffic, in bytes. \MCSI{} rows are derived from
the parameter choices as described in \Cref{sec:param-sets}. ML-KEM figures are the
encapsulation key and ciphertext sizes of FIPS 203, Table 3. Classic McEliece figures are
computed from the parameters and encoding rules of the round 4 specification, namely
$mt\lceil (n-mt)/8 \rceil$ bytes for the public key and $\lceil mt/8 \rceil$ bytes for the
ciphertext. Traffic is the total number of bytes exchanged in one session, counting both
directions. For \MCSI{} the static public keys are distributed out of band and are not part
of session traffic, whereas for the two key encapsulation mechanisms the encapsulation key is
transmitted during the session and is counted; this favours \MCSI{} in the traffic column and
should be kept in mind.}
\label{tab:sizes}
\begin{tabular}{llrrl}
\toprule
Scheme & Parameter set & Public key (B) & Session traffic (B) & Source \\
\midrule
\MCSI{}            & \MCSI-512          & 64        & 192  & derived, \cite{CSIDH2018} \\
\MCSI{}            & \MCSI-1024         & 128       & 320  & derived, \cite{CSIDH2018} \\
\MCSI{}            & \MCSI-4096         & 512       & 1088 & derived, \cite{SQALE2022} \\
\midrule
ML-KEM             & ML-KEM-512         & 800       & 1568 & \cite{FIPS203} \\
ML-KEM             & ML-KEM-768         & 1184      & 2272 & \cite{FIPS203} \\
ML-KEM             & ML-KEM-1024        & 1568      & 3136 & \cite{FIPS203} \\
\midrule
Classic McEliece   & mceliece348864     & 261\,120  & 261\,216 & derived, \cite{ClassicMcEliece2022} \\
Classic McEliece   & mceliece6688128    & 1\,044\,992 & 1\,045\,200 & derived, \cite{ClassicMcEliece2022} \\
\bottomrule
\end{tabular}
\end{table}

\Cref{fig:sizes} shows the public key column on a logarithmic scale, which is the only scale
on which the three families can be drawn together.

\begin{figure}[t]
\centering
\begin{tikzpicture}
\begin{axis}[
  width=0.86\textwidth,
  height=6.2cm,
  ybar,
  bar width=15pt,
  ymode=log,
  log basis y={10},
  ymin=30, ymax=100000000,
  ylabel={Static public key (bytes, log scale)},
  symbolic x coords={M512,M1024,M4096,K512,K768,K1024,CM1,CM2},
  xtick=data,
  xticklabels={%
    \ensuremath{\mathrm{MCSI}}-512, \ensuremath{\mathrm{MCSI}}-1024, \ensuremath{\mathrm{MCSI}}-4096,
    ML-KEM-512, ML-KEM-768, ML-KEM-1024,
    mceliece348864, mceliece6688128},
  x tick label style={rotate=35, anchor=east, font=\footnotesize},
  y tick label style={font=\footnotesize},
  ylabel style={font=\small},
  nodes near coords,
  point meta=explicit symbolic,
  nodes near coords style={font=\tiny, rotate=90, anchor=west},
  enlarge x limits=0.08,
  grid=major,
  grid style={gray!25},
]
\addplot[fill=gray!35, draw=black!70] coordinates {
  (M512,64)  [64]
  (M1024,128) [128]
  (M4096,512) [512]
  (K512,800) [800]
  (K768,1184) [1\,184]
  (K1024,1568) [1\,568]
  (CM1,261120) [261\,120]
  (CM2,1044992) [1\,044\,992]
};
\end{axis}
\end{tikzpicture}
\caption{Static public key sizes from \Cref{tab:sizes}. The vertical axis is logarithmic. The
figure shows sizes only; it is not a comparison of speed, of security level or of the
functionality the schemes provide.}
\label{fig:sizes}
\end{figure}
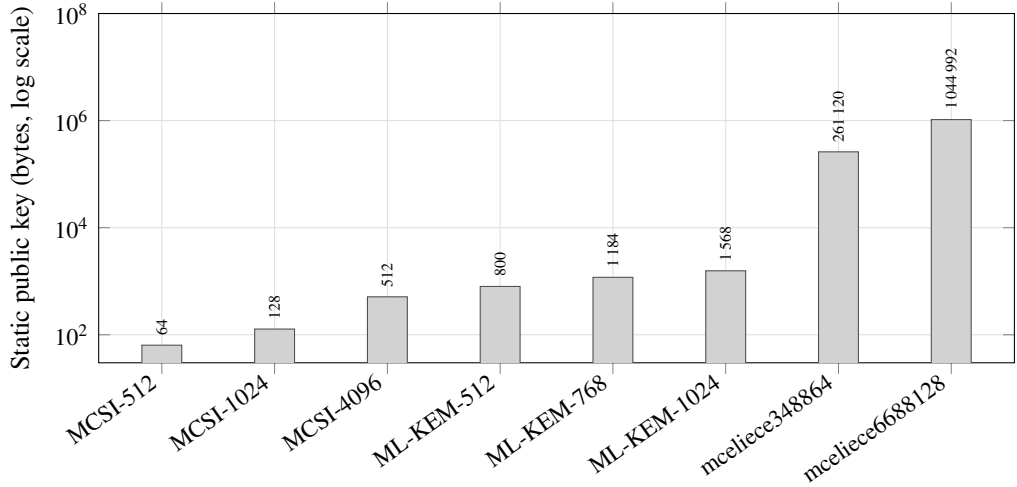

\subsection{What the protocol costs}
\label{sec:measured-cost}

We report two kinds of cost. Counts of operations in $\Fp$ do not depend on the machine, the
compiler or how busy the host is, so they are the figures to compare against another
implementation of the same protocol. Wall clock timings depend on all three, so we give the
machine, the compiler and the flags with them and treat them as an order of magnitude rather
than a benchmark.

\Cref{tab:opcount} gives the operation counts, as medians over $25$ runs with freshly sampled
keys. The column $M$ counts multiplications in $\Fp$, squarings included and the
multiplications performed inside the exponentiations included as well. The columns $I$ and
$L$ count inversions and Legendre symbols, each of which our implementation performs as one
exponentiation, costing $767$ and $766$ multiplications respectively; together they account
for under two percent of $M$.

\begin{table}[t]
\centering
\caption{Operations in $\Fp$ per step of \MCSI{}-512, medians over $25$ runs with fresh keys.
$M$ counts multiplications, squarings and the multiplications inside the exponentiations
included. $I$ counts inversions and $L$ Legendre symbols. The counts for anything containing
a group action vary between runs, because the evaluation samples random points and discards
those whose order does not let it act; the spread over the $25$ runs is about $\pm 15$ percent
for one evaluation. Reproduced by \texttt{make opcount}.}
\label{tab:opcount}
\begin{tabular}{lllrrr}
\toprule
Step & Group actions & Validations & $M$ & $I$ & $L$ \\
\midrule
\textsc{Validate}          & 0 & 1 &   270\,200 & 0 & 0 \\
one evaluation of the action & 1 & 0 &   690\,000 & 1 & 14.5 \\
\midrule
\textsc{PeerCtx}           & 1 & 1 &   885\,700 & 1 & 13.6 \\
\textsc{Init}              & 1 & 0 &   698\,000 & 1 & 15.2 \\
\textsc{Resp}              & 4 & 1 & 2\,958\,200 & 4 & 57.3 \\
\textsc{Fin}               & 3 & 1 & 2\,227\,100 & 3 & 43.6 \\
\midrule
one session, both parties  & 8 & 2 & 5\,836\,700 & 8 & 111.6 \\
\bottomrule
\end{tabular}
\end{table}

The counts follow the structure of \Cref{alg:peerctx,alg:session} exactly, which is the first thing they are
good for. \textsc{Init} is one evaluation of the action. \textsc{Resp} validates the received
element and then evaluates the action four times, for its own ephemeral element and for
$Z_{\mathrm{ee}}$, $Z_{\mathrm{es}}$ and $Z_{\mathrm{se}}$. \textsc{Fin} validates and
evaluates three times. \textsc{PeerCtx} validates the peer's static key and evaluates once,
and it runs once per peer rather than once per session. A session between two parties who
already hold each other's contexts therefore costs eight evaluations of the action and two
validations in total, which is the factor of two over unauthenticated Diffie--Hellman that
\Cref{sec:cost} predicts.

\Cref{tab:timing} gives wall clock timings. The machine is an Intel Xeon at a nominal
$2.80$ GHz inside a shared virtual machine, running Linux 6.18 and glibc 2.39, and the
compiler is gcc 13.3 at \texttt{-O2 -std=c11}. Each row is the median of $41$ samples, with
the first and third quartiles beside it.

\begin{table}[t]
\centering
\caption{Wall clock cost of \MCSI{}-512 with the reference implementation. Median, first and
third quartile over $41$ samples, on an Intel Xeon at a nominal $2.80$ GHz in a shared virtual
machine, gcc 13.3 at \texttt{-O2}. This is plain portable C with no assembly and no attempt at
a fast implementation, on a machine we do not have exclusive use of. It is a first
measurement of this protocol, not a competitive benchmark, and \Cref{sec:third-party} gives
the figures other authors report for the underlying action on hardware of their own.
Reproduced by \texttt{make bench}.}
\label{tab:timing}
\begin{tabular}{lrrr}
\toprule
Step & Median & $q_1$ & $q_3$ \\
\midrule
one multiplication in $\Fp$        & $0.18$ $\mu$s & $0.17$ & $0.19$ \\
one inversion or Legendre symbol   & $136$ $\mu$s  & $134$  & $139$ \\
one $512$-bit Montgomery ladder    & $1.16$ ms & $1.14$ & $1.18$ \\
one $587$-isogeny                  & $0.85$ ms & $0.84$ & $0.87$ \\
\midrule
key generation, one action         & $130$ ms & $121$ & $134$ \\
\textsc{Validate}                  & $53$ ms  & $52$  & $54$ \\
\textsc{PeerCtx}, once per peer    & $196$ ms & $195$ & $200$ \\
\textsc{Init}                      & $127$ ms & $119$ & $134$ \\
\textsc{Resp}                      & $557$ ms & $537$ & $586$ \\
\textsc{Fin}                       & $445$ ms & $435$ & $461$ \\
one session, both parties          & $1132$ ms & $1100$ & $1160$ \\
\midrule
authenticated encryption of the $64$-byte body & $4.2$ $\mu$s & $4.2$ & $4.3$ \\
\bottomrule
\end{tabular}
\end{table}

\begin{figure}[t]
\centering
\begin{tikzpicture}
\begin{axis}[
  width=0.80\textwidth,
  height=6.4cm,
  xbar,
  bar width=13pt,
  xmin=0, xmax=1620,
  xlabel={median wall clock time (ms)},
  symbolic y coords={AEAD,Validate,Init,KeyGen,PeerCtx,Fin,Resp,Session},
  ytick=data,
  y dir=reverse,
  yticklabels={%
    AEAD on the body,
    \textsc{Validate},
    \textsc{Init},
    key generation,
    \textsc{PeerCtx},
    \textsc{Fin},
    \textsc{Resp},
    whole session},
  y tick label style={font=\footnotesize},
  x tick label style={font=\footnotesize},
  xlabel style={font=\small},
  nodes near coords,
  point meta=explicit symbolic,
  nodes near coords style={font=\scriptsize, anchor=west},
  enlarge y limits=0.09,
  grid=major,
  grid style={gray!25},
]
\addplot[fill=gray!35, draw=black!70] coordinates {
  (0.0042,AEAD)   [0.004 ms, no group action]
  (53,Validate)   [53 ms, 1 validation]
  (127,Init)      [127 ms, 1 action]
  (130,KeyGen)    [130 ms, 1 action]
  (196,PeerCtx)   [196 ms, 1 action, 1 validation]
  (445,Fin)       [445 ms, 3 actions, 1 validation]
  (557,Resp)      [557 ms, 4 actions, 1 validation]
  (1132,Session)  [1132 ms, 8 actions, 2 validations]
};
\end{axis}
\end{tikzpicture}
\caption{Where the time goes in \MCSI{}-512, from the medians of
\Cref{tab:timing}. The label on each bar gives the measured time and the
number of class group action evaluations and validations that step performs.
The cost of the protocol is the count of group actions and nothing else: the
whole symmetric layer, the top bar, is four microseconds against a session of
more than a second, and is not visible at this scale. \textsc{PeerCtx} is
paid once per peer rather than once per session. Reproduced by \texttt{make bench}.}
\label{fig:cost}
\end{figure}
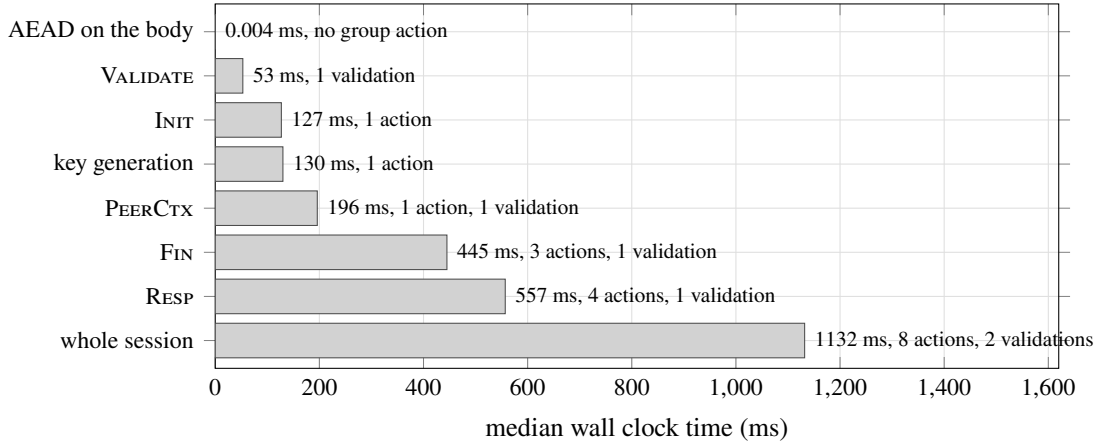

The two tables agree with each other, which is the second thing the counts are good for:
$690\,000$ multiplications at $0.18$ $\mu$s each is $124$ ms, against $130$ ms measured for
one evaluation of the action, and the remainder is the additions and the loop overhead. The
same source compiled with clang 18.1 at \texttt{-O2} runs about a quarter faster, $848$ ms for
a session, and gcc at \texttt{-O3 -march=native} gains about five percent, $1076$ ms. We
report the gcc \texttt{-O2} figures because they are the ones the default build produces.

\subsection{Memory}
\label{sec:memory}

The report this paper grew out of claimed a memory figure that no measurement
supported, so we are careful here to say which numbers are exact, which are
measured, and which depend on the machine.

The objects a party has to keep are exact and follow from the parameter set. A private key
is $74$ signed exponents, so $74$ bytes; a public key is one field element, so $64$; a key
pair is $144$ bytes with the padding the compiler inserts. The cached context of
\Cref{alg:peerctx}, which holds the static secret, both static public keys, the shared static
curve and the two masking keys, is $304$ bytes and is kept once per peer. The per session
state is $170$ bytes and lives only until the session finishes. Nothing else persists.

The working memory is dominated by one thing, the scratch space of the isogeny routine. It
holds the $s = (\ell - 1)/2$ multiples of the kernel point and the two arrays of running
products used to avoid an inversion, so at the largest degree $\ell = 587$ it needs
$293 \cdot 128$ bytes for the multiples and $2 \cdot 294 \cdot 64$ for the products,
$75\,136$ bytes in all. That is the price of evaluating V\'elu's formulae without an
inversion, and it is fixed by the parameter set rather than by the implementation.

The measured figures agree with the design, and \texttt{make memory} reproduces them. The
library calls no allocator at all, so the heap use of a complete session is zero bytes; the deepest the stack goes during a session,
measured with valgrind's massif, is $7\,680$ bytes. The compiled library is $25\,897$ bytes
of code. A party therefore needs on the order of $100$ kilobytes in total, almost all of it
the isogeny scratch, and needs no dynamic memory, which is the property that matters if this
were ever to run somewhere small. Peak resident set for the whole process is about
$2.4$ megabytes, but that is the C library and the loader rather than \MCSI{}, and we report
it only because a reader will otherwise wonder.

\subsection{Third party timings for the underlying group action}
\label{sec:third-party}

The cost of \MCSI{} is four evaluations of the class group action per party per session, plus
one cached evaluation per peer, plus validation of each received element. Published timings
for the action itself therefore bound what \MCSI{} can achieve. We quote them with their
platforms attached, and we emphasise that figures measured on different processors by
different authors with different compilers are not directly comparable with each other.

Castryck et al.~\cite{CSIDH2018} report $40.8$ ms for one evaluation of the
CSIDH-512 action and $2.1$ ms for validation, measured on an Intel Skylake i5 clocked at
$3.5$ GHz with a proof of concept implementation that is not constant time. Later work has
improved on this and also, in the constant time setting, paid for the improvement: Meyer and
Reith~\cite{MeyerReith2018} gave faster formulae, Meyer, Campos and
Reith~\cite{MeyerCamposReith2019} and Onuki et al.~\cite{OnukiEtAl2019} gave constant time
algorithms, and CTIDH~\cite{CTIDH2021} introduced a different key space with a substantially
faster constant time evaluation. At the much larger parameters of the last row of
\Cref{tab:params}, Ch\'avez-Saab et al.~\cite{SQALE2022} report their own measurements on an
Intel Core i7-6700K at $4.00$ GHz.

Taking the CSIDH-512 figure at face value, a session of \MCSI{}-512 costs at least four
evaluations, or on the order of $0.16$ s per party on that platform, plus validation. Our own
implementation is about three times slower per evaluation than the figure Castryck et al.\
report, $130$ ms against $40.8$ ms, on a slower and shared machine and with none of the
optimisations their proof of concept already contains. We take the gap as the expected
distance between a first portable implementation and a tuned one, and we do not read it as
evidence about the protocol. Either way the conclusion is the same, and it does not depend on
which of the two numbers is used: \MCSI{} is far slower than a lattice based key
encapsulation, and no arrangement of the protocol changes that. The case for a scheme in this
family rests on the size column of \Cref{tab:sizes} and on the algebraic diversity of the
underlying assumption, not on speed, and we state it that way rather than the other way
round.

\subsection{The symmetric layer costs nothing by comparison}
\label{sec:symmetric-cost}

Each party performs two hash calls and two authenticated encryption operations on payloads of
$L$ bytes, where $L$ is at most $512$. \Cref{tab:timing} puts a number on what that costs:
$4.2$ $\mu$s for the authenticated encryption of a $64$-byte body, against $1132$ ms for
the session that contains four of them. The whole symmetric layer is therefore about one part
in $10^{5}$ of a session, and the masking is free in any sense that matters. The cost of
\MCSI{} relative to an unauthenticated key exchange over the same group action is the factor
of two in action evaluations identified in \Cref{sec:cost}, and essentially nothing else.

The same measurement gives the size of the effect that motivates the design in the first
place. Rejecting a message whose tag does not verify takes $2.4$ $\mu$s, because the
implementation stops at the tag and never generates the keystream, against $130$ ms for the
evaluation of the group action that a protocol without the masking layer would have performed
before it could tell. That is a factor of about $50\,000$, and it is what an unauthenticated
responder spends on behalf of anybody who can reach it.
\subsection{What the implementation shows about side channels}
\label{sec:measured-leak}

\Cref{sec:side-channels} states that we make no claim of resistance to timing attacks. With
an implementation in hand we can say something sharper than that, and we prefer a measurement
to a disclaimer.

We audited the code with the technique Langley published as
ctgrind~\cite{Langley2010ctgrind}: secret buffers are handed to valgrind's memcheck as
uninitialised memory, so that the tool reports every conditional branch, every memory index
and every division whose outcome depends on them. A report is an actual secret dependent
control flow found by the tool, not a judgement of ours. A small number of values are derived
from secrets but are public by design, the clearest being whether an authenticated decryption
succeeded, since the peer observes the answer either way; those are marked in the source and
excluded from the report. \Cref{tab:ctaudit} summarises what the audit found, and
\texttt{make ctcheck} reproduces it.

\begin{table}[t]
\centering
\caption{Result of the constant time audit of the reference implementation, using valgrind
3.22 on the build of \Cref{tab:timing}. The audit is reproducible from the repository.}
\label{tab:ctaudit}
\begin{tabular}{ll}
\toprule
Layer & Secret dependent control flow found \\
\midrule
arithmetic in $\Fp$ & none \\
SHA-256, HMAC and the authenticated encryption & none beyond the accept or reject decision \\
Montgomery ladder & the bit length of the scalar \\
evaluation of the class group action & the exponent vector, at three branches \\
\bottomrule
\end{tabular}
\end{table}

The first two rows were not free. The Montgomery multiplication originally ended with a
branch on the top word of its accumulator, and the Legendre symbol returned early when its
argument was zero; the first is dead code for a prime below $2^{511}$ and is now removed under
a compile time assertion on the size of the prime, and the second is now folded in with a
mask. The tag comparison in the authenticated decryption accumulates the difference of all
sixteen bytes before it is looked at, so it does not stop at the first mismatch.

The last row is not an oversight in the coding, and no amount of care with branches would
remove it. The evaluation of \Cref{sec:effective-action} performs a number of isogeny steps
equal to the one-norm of the key and decides which isogeny to take from the sign of each
exponent, so its running time follows the key by construction. This is exactly what the
constant time algorithms cited in \Cref{sec:side-channels} are for, and we did not implement
any of them.

\Cref{fig:leak} shows how much that costs in practice. With keys whose exponents all have the
same magnitude the time is a straight line in the one-norm, from $0.13$ ms for the all-zero
key to $256$ ms for the largest key the key space allows, a factor of about two thousand. On
$200$ keys drawn the way \textsc{KeyGen} draws them the one-norm and the running time have a
Pearson correlation of $0.90$ and a Spearman correlation of $0.90$, and a least squares fit
gives $0.68$ ms per isogeny step with a residual standard deviation of $4.5$ ms.

The last two numbers are the ones with a security reading. The one-norm of a uniform key in
$\{-5,\dots,5\}^{74}$ has standard deviation $13.77$, and our sample of $200$ keys gave
$13.84$. A single timing observation locates it to within a standard deviation of
$4.5 / 0.68$, that is $6.6$ steps, so one timing of one evaluation recovers about one bit
about the one-norm of the key. Host noise can only weaken an observed correlation, so this is
a lower bound on what a quieter machine would reveal, and an adversary able to separate the
individual isogeny steps learns much more than the one-norm.

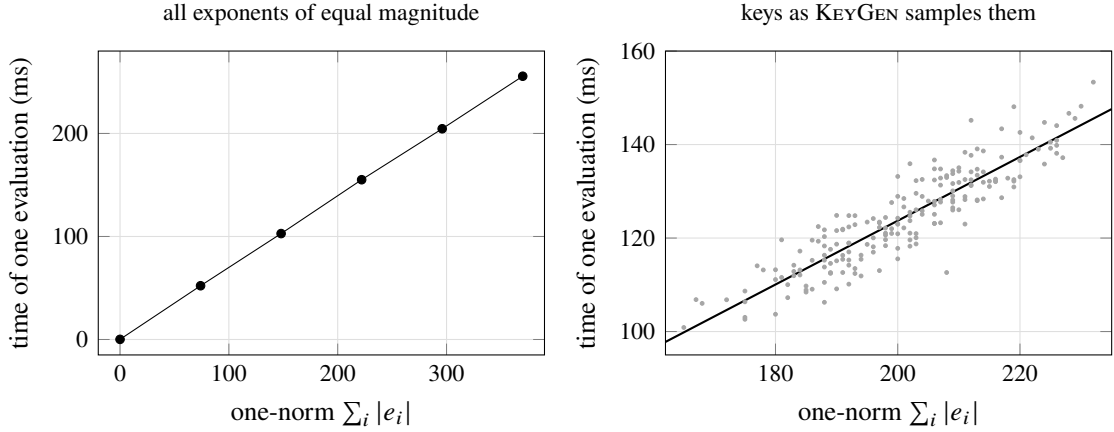
\begin{figure}[t]
\centering
\begin{tikzpicture}
\begin{groupplot}[
  group style={group size=2 by 1, horizontal sep=1.6cm},
  width=0.47\textwidth,
  height=5.6cm,
  xlabel style={font=\small},
  ylabel style={font=\small},
  tick label style={font=\footnotesize},
  grid=major,
  grid style={gray!25},
]
\nextgroupplot[
  xlabel={one-norm $\sum_i |e_i|$},
  ylabel={time of one evaluation (ms)},
  xmin=-20, xmax=390, ymin=-15, ymax=280,
  title={\footnotesize all exponents of equal magnitude},
]
\addplot[black, mark=*, mark size=1.6pt] coordinates { (0,0.13) (74,52.17) (148,102.81) (222,155.09) (296,204.56) (370,255.56) };

\nextgroupplot[
  xlabel={one-norm $\sum_i |e_i|$},
  ylabel={time of one evaluation (ms)},
  xmin=162, xmax=235, ymin=95, ymax=160,
  title={\footnotesize keys as \textsc{KeyGen} samples them},
]
\addplot[only marks, mark=*, mark size=0.7pt, gray!70] coordinates {
  (165,100.87) (167,106.83) (168,106.03) (172,106.80) (175,102.62) (175,103.04) (175,106.36) (175,108.66)
  (177,114.06) (178,113.18) (180,103.70) (180,111.12) (180,113.22) (181,111.57) (181,119.62) (182,107.25)
  (182,110.03) (183,111.94) (183,112.86) (183,114.19) (184,112.10) (184,113.17) (184,117.21) (185,108.44)
  (185,108.88) (185,109.74) (186,109.06) (186,115.24) (186,119.55) (187,119.32) (187,122.46) (188,106.26)
  (188,110.52) (188,113.80) (188,115.27) (188,116.33) (188,118.31) (188,120.31) (188,121.76) (189,112.49)
  (189,112.50) (189,112.60) (189,114.58) (189,117.59) (190,109.13) (190,115.13) (190,118.69) (190,121.60)
  (190,124.84) (191,112.63) (191,115.45) (191,116.30) (191,121.76) (191,121.90) (192,110.14) (192,114.20)
  (192,115.30) (192,117.03) (192,118.91) (192,118.92) (192,122.38) (192,124.80) (193,112.42) (193,113.48)
  (193,122.86) (193,124.79) (194,111.89) (194,115.46) (194,116.34) (195,118.65) (195,123.42) (196,116.98)
  (196,118.14) (196,120.04) (196,124.17) (197,113.04) (197,122.27) (197,123.25) (197,123.42) (198,118.37)
  (198,119.65) (198,120.63) (198,121.04) (198,122.15) (198,124.40) (199,120.98) (199,122.00) (199,127.60)
  (200,115.55) (200,117.80) (200,120.91) (200,122.96) (200,123.04) (200,124.25) (200,128.92) (200,133.18)
  (201,122.17) (201,126.66) (201,128.46) (202,118.26) (202,119.60) (202,121.05) (202,123.71) (202,124.71)
  (202,125.52) (202,135.90) (203,118.74) (203,120.06) (203,120.95) (203,126.04) (203,129.57) (203,132.28)
  (204,125.07) (204,125.83) (204,128.91) (204,132.54) (205,127.92) (206,123.12) (206,123.13) (206,125.08)
  (206,127.08) (206,127.86) (206,132.42) (206,134.80) (206,136.70) (207,123.27) (207,127.69) (207,128.28)
  (207,131.32) (207,132.83) (207,134.78) (208,112.63) (208,132.96) (208,133.35) (209,125.99) (209,126.48)
  (209,127.68) (209,128.13) (209,129.01) (209,130.09) (209,131.66) (209,133.78) (209,134.40) (210,126.05)
  (210,131.90) (210,134.30) (210,135.04) (211,122.99) (211,128.22) (211,128.88) (211,132.52) (211,138.35)
  (212,131.22) (212,131.38) (212,132.46) (212,134.16) (212,137.70) (212,145.19) (213,128.00) (213,128.16)
  (213,133.45) (213,134.73) (213,139.07) (214,128.37) (214,131.93) (214,134.15) (214,138.63) (215,131.03)
  (215,132.16) (216,132.06) (216,132.25) (216,132.60) (217,128.63) (217,137.30) (217,143.37) (218,132.81)
  (219,130.93) (219,132.17) (219,132.50) (219,148.09) (220,133.10) (220,136.56) (220,136.63) (220,142.60)
  (221,137.82) (222,141.44) (223,139.00) (224,135.83) (224,144.74) (225,139.21) (225,140.55) (226,138.13)
  (226,139.83) (226,140.91) (226,144.05) (227,137.19) (228,146.68) (229,145.60) (230,148.20) (232,153.35)
};
\addplot[black, thick] coordinates { (162,97.78) (235,147.59) };
\end{groupplot}
\end{tikzpicture}
\caption{Running time of one evaluation of the class group action against the one-norm of the
private key, on the machine of \Cref{tab:timing}. Each point is the fastest of three runs of
the same key, which removes most of the noise the host contributes. Left: six keys whose
exponents all have the same magnitude, from the all-zero key to the largest key the key space
allows. Right: $200$ keys sampled the way \textsc{KeyGen} samples them, with the least squares
line through them. The measurement is of the reference implementation, which is not constant
time; a constant time evaluation would produce a horizontal line in both panels. Reproduced
by \texttt{make leak}.}
\label{fig:leak}
\end{figure}

The correlation of \Cref{fig:leak} says that the timing carries information about the key.
It does not say how many measurements an adversary needs, and that is the question a
practitioner asks. To answer it we ran the leakage assessment that the side channel
literature uses: time two classes of inputs, apply Welch's $t$-test to the two sets of
timings, and read a difference as detected when $|t|$ passes $4.5$. The classes are
interleaved at random rather than measured in two blocks, so a machine that slows down halfway
through the run cannot masquerade as a leak. \Cref{tab:tvla} gives the outcome and
\Cref{fig:tvla} shows how $|t|$ grows with the number of measurements.

\begin{table}[t]
\centering
\caption{Leakage assessment of \MCSI{}-512 on the machine of \Cref{tab:timing}. Class A is
the fixed input and class B the varying one; the two are interleaved at random. A difference
is reported when $|t|$ over all the samples passes $4.5$. The last column is the ratio of the
two sample variances, which catches a secret that changes the spread of the timings without
moving their centre. Reproduced by \texttt{make tvla}.}
\label{tab:tvla}
\begin{tabular}{llrrrl}
\toprule
Quantity & Two classes & $n$ & $\Delta$ median & $t$ & Verdict \\
\midrule
AEAD encrypt   & fixed against random key & $200\,000$ & $0.0$ ns & $0.77$ & no difference \\
AEAD decrypt   & fixed against random key & $200\,000$ & $0.0$ ns & $0.06$ & no difference \\
AEAD decrypt   & accept against reject    & $200\,000$ & $-904$ ns & $249$ & difference \\
\textsc{Validate} & two valid public keys & $400$ & $-0.03$ ms & $-0.12$ & no difference \\
group action   & fixed against random key & $400$ & $-4.1$ ms & $11.3$ & difference \\
\midrule
group action   & the same key twice       & $200$ & $+0.1$ ms & $-1.4$ & no difference \\
group action   & two keys, one-norm $180$ each & $200$ & $+0.4$ ms & $-0.8$ & no difference \\
group action   & one-norm $180$ against $182$ & $200$ & $+2.0$ ms & $-4.1$ & no difference \\
group action   & one-norm $180$ against $185$ & $200$ & $+3.2$ ms & $-12.5$ & difference \\
group action   & one-norm $180$ against $190$ & $200$ & $+5.1$ ms & $-17.5$ & difference \\
group action   & one-norm $180$ against $200$ & $200$ & $+8.3$ ms & $-18.5$ & difference \\
group action   & one-norm $180$ against $220$ & $200$ & $+23.7$ ms & $-46.4$ & difference \\
\bottomrule
\end{tabular}
\end{table}

Four things in that table are worth drawing out.

The symmetric layer shows nothing. After two hundred thousand measurements per class the
difference between a fixed key and random keys is zero at the resolution of the clock and
$|t|$ is below one, for both directions of the authenticated encryption. This is the
statistical counterpart of the audit result, and the two agree.

The accept or reject row is the control. The rejection path is about $900$ nanoseconds
faster, because it returns before generating the keystream, and the test sees that
immediately. We are not treating it as a finding: which of the two happened is exactly what
the protocol tells the peer in its next message. Its value here is that it demonstrates the
test can find a difference that is present, which is what makes the two rows above it worth
reading.

\textsc{Validate} shows nothing between two valid public keys. That is expected rather than
reassuring, since the routine walks the same list of $74$ primes either way, and it says
nothing about the timing of a rejection.

The last block is the one with teeth. Two keys, each fixed and used for every measurement in
its class, are told apart from $200$ timings once their one-norms differ by five out of a
possible $370$. A difference of two sits right at the threshold. Two keys with the same
one-norm are not separated at this sample size, which is consistent with the model of
\Cref{fig:leak} in which the running time follows the number of isogeny steps, though the
cost also depends on which primes carry the weight, since a $587$-isogeny costs two hundred
times a $3$-isogeny.

One methodological note, because it would otherwise mislead. The standard fixed against
random test, the fifth row, is a weak instrument here: it compares means, and a fixed key
drawn at random has a typical one-norm, so its mean sits near the mean of the random class.
What separates the two classes is the spread, which is why the variance ratio is close to
seven while $t$ is only $11.3$. An implementation of a group action should be assessed with
pairs of fixed keys, as in the last block, and not only with the usual fixed against random
pair.

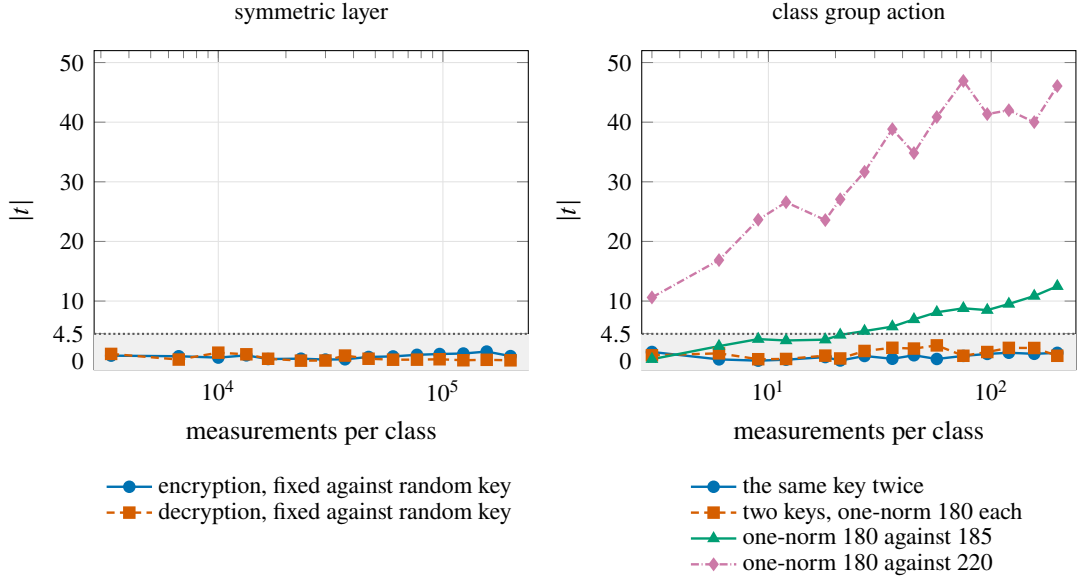
\begin{figure}[t]
\centering
\definecolor{mcsiblue}{RGB}{0,114,178}
\definecolor{mcsiverm}{RGB}{213,94,0}
\definecolor{mcsigreen}{RGB}{0,158,115}
\definecolor{mcsipurple}{RGB}{204,121,167}
\begin{tikzpicture}
\begin{groupplot}[
  group style={group size=2 by 1, horizontal sep=1.5cm},
  width=0.46\textwidth,
  height=5.8cm,
  xmode=log,
  ymin=-1.5, ymax=52,
  ytick={0,4.5,10,20,30,40,50},
  yticklabels={$0$,$4.5$,$10$,$20$,$30$,$40$,$50$},
  ylabel={$|t|$},
  xlabel={measurements per class},
  xlabel style={font=\small},
  ylabel style={font=\small},
  tick label style={font=\footnotesize},
  legend style={font=\footnotesize, at={(0.5,-0.30)}, anchor=north,
                draw=none, fill=none, legend columns=1, row sep=-2pt},
  legend cell align=left,
  grid=major,
  grid style={gray!22},
]
\nextgroupplot[
  xmin=2800, xmax=240000,
  title={\footnotesize symmetric layer},
]
\fill[black!6] (axis cs:2800,-1.5) rectangle (axis cs:240000,4.5);
\addplot[mcsiblue, solid, line width=0.9pt, mark=*, mark size=1.9pt, mark options={solid, fill=mcsiblue}] coordinates { (3333,0.845) (6666,0.745) (9999,0.517) (13332,0.894) (16665,0.279) (23331,0.339) (29997,0.189) (36663,0.230) (46662,0.627) (59994,0.718) (76659,0.985) (96657,1.121) (123321,1.204) (156651,1.530) (199980,0.770) };
\addlegendentry{encryption, fixed against random key}
\addplot[mcsiverm, densely dashed, line width=0.9pt, mark=square*, mark size=1.9pt, mark options={solid, fill=mcsiverm}] coordinates { (3333,1.137) (6666,0.190) (9999,1.344) (13332,1.071) (16665,0.326) (23331,0.008) (29997,0.032) (36663,0.858) (46662,0.330) (59994,0.166) (76659,0.173) (96657,0.241) (123321,0.076) (156651,0.149) (199980,0.063) };
\addlegendentry{decryption, fixed against random key}
\draw[black!65, densely dotted, line width=0.8pt]
  (axis cs:2800,4.5) -- (axis cs:240000,4.5);

\nextgroupplot[
  xmin=2.7, xmax=240,
  title={\footnotesize class group action},
]
\fill[black!6] (axis cs:2.7,-1.5) rectangle (axis cs:240,4.5);
\addplot[mcsiblue, solid, line width=0.9pt, mark=*, mark size=1.9pt, mark options={solid, fill=mcsiblue}] coordinates { (3,1.464) (6,0.219) (9,0.006) (12,0.167) (18,0.604) (21,0.027) (27,0.788) (36,0.321) (45,0.905) (57,0.287) (75,0.803) (96,1.136) (120,1.331) (156,1.150) (198,1.318) };
\addlegendentry{the same key twice}
\addplot[mcsiverm, densely dashed, line width=0.9pt, mark=square*, mark size=1.9pt, mark options={solid, fill=mcsiverm}] coordinates { (3,0.902) (6,1.231) (9,0.248) (12,0.292) (18,0.848) (21,0.373) (27,1.632) (36,2.165) (45,2.034) (57,2.565) (75,0.825) (96,1.454) (120,2.150) (156,2.131) (198,0.824) };
\addlegendentry{two keys, one-norm 180 each}
\addplot[mcsigreen, solid, line width=0.9pt, mark=triangle*, mark size=1.9pt, mark options={solid, fill=mcsigreen}] coordinates { (3,0.249) (6,2.440) (9,3.620) (12,3.391) (18,3.527) (21,4.321) (27,4.966) (36,5.717) (45,6.927) (57,8.121) (75,8.787) (96,8.482) (120,9.512) (156,10.848) (198,12.490) };
\addlegendentry{one-norm 180 against 185}
\addplot[mcsipurple, densely dashdotted, line width=0.9pt, mark=diamond*, mark size=1.9pt, mark options={solid, fill=mcsipurple}] coordinates { (3,10.609) (6,16.869) (9,23.644) (12,26.579) (18,23.579) (21,27.066) (27,31.669) (36,38.833) (45,34.826) (57,40.867) (75,46.916) (96,41.356) (120,42.002) (156,40.009) (198,46.059) };
\addlegendentry{one-norm 180 against 220}
\draw[black!65, densely dotted, line width=0.8pt]
  (axis cs:2.7,4.5) -- (axis cs:240,4.5);
\end{groupplot}
\end{tikzpicture}
\caption{Welch's $t$ statistic against the number of measurements per class,
on the machine of \Cref{tab:timing}, with the same vertical scale in both
panels and the detection threshold of $4.5$ marked. The horizontal axis is
logarithmic and its range differs between the panels, which is the point of
the figure: the symmetric layer is still under the threshold after two hundred
thousand measurements per class, while two keys of the class group action
cross it after a few dozen. The shaded band is the region in which no
difference is detected. Each curve is drawn at sample counts spaced evenly
along the logarithmic axis. The control, an authenticated decryption timed
with a valid tag against a corrupted one, is left out because it reaches
$|t| = 249$ and would leave the panel; it appears in \Cref{tab:tvla} and shows
the test detecting a difference that is really there. Reproduced by
\texttt{make tvla}.}
\label{fig:tvla}
\end{figure}

For \MCSI{} specifically this is worse than it would be for a scheme that used its static key
once. The static secret is used in \textsc{PeerCtx} when the context for a peer is built, and
then again in every session: \textsc{Resp} evaluates the action with it to obtain
$Z_{\mathrm{es}}$ and \textsc{Fin} evaluates the action with it to obtain $Z_{\mathrm{se}}$.
An adversary who can time many sessions with the same peer is timing the same static key over
and over, and averaging removes the noise that limits a single observation. Any deployment
therefore needs a constant time evaluation of the action, not as a refinement but as a
precondition.

\section{Related work}
\label{sec:related}

\subsection{Commutative isogeny actions}

The idea of using the action of an ideal class group on a set of elliptic curves as a
cryptographic primitive is due independently to Couveignes~\cite{Couveignes2006}, who
formulated it abstractly as a hard homogeneous space, and to Rostovtsev and
Stolbunov~\cite{RostovtsevStolbunov2006}. Both worked with ordinary curves, where the action
is available but slow, since the small primes that make the action cheap must divide the order
of the curve and cannot be chosen freely. Castryck, Lange, Martindale, Panny and
Renes~\cite{CSIDH2018} moved the construction to supersingular curves over $\Fp$, where
\Cref{prop:endo-commutative} supplies a commutative endomorphism ring and where the prime can
be chosen so that $p+1$ is smooth. The result, CSIDH, is the group action we instantiate.

Subsequent work has refined the same action rather than replaced it. CSURF~\cite{CSURF2020}
works on the surface of the volcano and adds a $2$-isogeny direction.
CSI-FiSh~\cite{CSIFiSh2019} computed the class group and a reduced relation lattice for the
CSIDH-512 parameters, which converts the restricted action into a genuine effective group
action and enables signatures; SCALLOP~\cite{SCALLOP2023} attacks the scaling problem that
this computation runs into at larger parameters. On the implementation
side~\cite{MeyerReith2018,MeyerCamposReith2019,OnukiEtAl2019,CTIDH2021} give progressively
faster and constant time evaluations, \cite{SqrtVelu2020} reduces the asymptotic cost of a
single large degree step, and \cite{SQALE2022} explores the very large parameters that the
quantum cryptanalysis of \Cref{sec:security-levels} may require. Alamati, De Feo, Montgomery
and Patranabis~\cite{AlamatiEtAl2020} abstracted the interface, and it is their vocabulary of
effective and restricted effective group actions that we use in \Cref{def:ega}.

\subsection{SIDH, its break, and masking of torsion data}

The other branch of isogeny based cryptography starts from Jao and De
Feo~\cite{JaoDeFeo2011,DeFeoJaoPlut2014}, who work with supersingular curves over $\Fpsq$
where the endomorphism ring is a quaternion order and the isogeny graph is not commutative.
To make a Diffie--Hellman analogue work in that setting, SIDH publishes the images of a
torsion basis under the secret isogeny. Galbraith, Petit, Shani and Ti~\cite{GPST2016} showed
early that this auxiliary data is dangerous when static keys are reused. In 2022 Castryck and
Decru~\cite{CastryckDecru2023} turned the torsion images into a full key recovery attack, and
Maino et al.~\cite{Maino2023} and Robert~\cite{Robert2023} generalised it, Robert removing the
last restrictions and giving a polynomial time algorithm.

Several responses attempt to keep the SIDH structure while hiding the torsion data. The most
developed is M-SIDH and MD-SIDH by Fouotsa, Moriya and Petit~\cite{MSIDH2023}, which scales
the torsion images by a secret value, respectively also hides the degree, and shows what
parameter growth this costs. We draw attention to this line of work because the phrase
``masking torsion points'' has been used for it, and our use of the word masking is a
different thing altogether: we transmit no torsion data, so there is nothing of that kind to
mask, and what our blinding layer hides is the identity of an ephemeral curve. The distinction
is spelled out in \Cref{sec:torsion-attacks}. The commutative branch, including \MCSI{}, was
unaffected by the 2022 attacks for the structural reason that it publishes no torsion images
and fixes no isogeny degree.

Isogeny based cryptography over $\Fpsq$ continues in other directions, notably the signature
scheme SQISign~\cite{SQISign2020}, which uses the quaternion endomorphism ring constructively
rather than trying to hide it.

\subsection{Authenticated key exchange from group actions}

\MCSI{} is not the first authenticated key exchange built on the CSIDH action, and we do not
claim the idea. De Kock, Gj{\o}steen and Veroni~\cite{deKock2021} give a two message
authenticated key exchange from hard homogeneous spaces with a tight security proof, and
Kawashima, Takashima, Aikawa and Takagi~\cite{Kawashima2021} give one exploiting the random
self-reducibility of the CSIDH problem. Both are proved secure in strong models, which
\Cref{thm:sk-ind} is not; both send their ephemeral curves in the clear, which \MCSI{} does
not. Our contribution relative to that work is narrow and can be stated in one sentence: we
add a blinding layer keyed by the static-static shared value, we prove what it buys
(\Cref{thm:blinding} and \Cref{thm:integrity}), and we are explicit about what it costs, which
is one extra cached group action evaluation per peer, the loss of forward secrecy for the
blinding itself, and the requirement that static keys be distributed authentically before any
session begins.

The general study of what can be built from cryptographic group actions is
in~\cite{AlamatiEtAl2020}, which also contains the group action analogues of the standard
Diffie--Hellman assumptions that we use in \Cref{sec:hard-problems}. The observation that the
decisional assumption fails for class group actions of non-prime discriminant is due to
Castryck, Sot\'akov\'a and Vercauteren~\cite{CSV2020}, and it is the reason our proofs are
phrased in the random oracle model over a computational assumption rather than in the standard
model over a decisional one.

\subsection{Hiding public key material in key exchange}

The pattern of encrypting a session's ephemeral public key under a value derived from long
term keys is standard outside isogenies. It appears in authenticated key exchange designs
that mix static and ephemeral Diffie--Hellman contributions, and the accounting of which
combination provides authentication and which provides forward secrecy goes back to the
analyses of Canetti and Krawczyk~\cite{CanettiKrawczyk2001} and LaMacchia, Lauter and
Mityagin~\cite{LaMacchia2007}. What we contribute here is not the pattern but its transfer to a
group action, where the arithmetic that makes the two parties agree is
\Cref{thm:class-action} rather than exponentiation, and where the proof cannot use a decisional
assumption for the reason just given.
\section{Conclusion and open problems}
\label{sec:conclusion}

We have specified \MCSI{}, a two message key exchange over the CSIDH class group action in
which each party's ephemeral public element is transmitted under an authenticated encryption
whose key is derived from the static-static shared value. We proved that the protocol is
correct with zero error (\Cref{thm:correctness}), that its session key is indistinguishable
from random against a passive adversary under the strong parallelisation assumption in the
random oracle model (\Cref{thm:sk-ind}), that the blinded ephemeral elements are hidden from
such an adversary (\Cref{thm:blinding}), and that an adversary holding neither static secret
cannot get a message accepted (\Cref{thm:integrity}). We also showed that a public byte
substitution layer of the kind sometimes proposed for constructions of this shape changes no
security notion at all (\Cref{lem:sbox}).

We also implemented the protocol twice, in portable C and in Python, checked the two against
each other by known answer vectors, and reported what a session costs in operations in $\Fp$,
which do not depend on the machine, in wall clock time, which does, and in memory, which the
protocol fixes almost entirely through the scratch space of the isogeny formulae
(\Cref{sec:measured-cost,sec:memory}). We audited that implementation for secret dependent
branches and measured what its timing gives away, and we report both results in full,
including the part that is bad news (\Cref{sec:measured-leak}).

It is as important to say what the paper does not establish. \Cref{sec:limits} lists six
items: security in the Canetti--Krawczyk or extended Canetti--Krawczyk models, security
against a fully active adversary, forward secrecy of the blinding, security against a
malicious peer with a valid static key, the concrete quantum security of the smaller parameter
sets, and resistance to implementation attacks. None of these is claimed anywhere in the
paper, and the first four are the natural next pieces of work.

Three further points deserve to be recorded, because they are the conclusions we would most
want a reader building something similar to take away.

First, the commutativity that makes a Diffie--Hellman analogue possible is a property of the
ideal class group acting on a torsor, and not a property of isogenies. Constructions that
iterate an isogeny step and hope that the two orders of composition agree do not have it, and
no correctness theorem of the kind in \Cref{sec:correctness} can be proved for them.

Second, a blinding layer inside a key exchange must be keyed by material that exists before
the session, or it is circular. Once that is accepted, the design space is small, and the
option we took has visible and stateable costs: static keys, an authentic distribution channel
for them, and no forward secrecy for the blinding.

Third, the choice of prime in a commutative isogeny scheme is not free and is not inherited
from elsewhere. \Cref{rem:2521} shows that the prime of a standard elliptic curve, chosen by
criteria that are entirely reasonable for discrete logarithms, admits no usable class group
action at all. The prime must be built for the action, as in \Cref{eq:prime-form}.

We close with the open problems in the order we would attack them. Proving security in a model
with ephemeral key reveal, most plausibly by adapting the techniques
of~\cite{deKock2021,Kawashima2021} to a protocol whose first message is encrypted, is the most
valuable. Constructing a blinding layer that is forward secret, or proving that none exists
without additional setup, is the most interesting. Producing a constant time implementation is the
most overdue. The one released with this paper is not one, and \Cref{sec:measured-leak} shows
what that costs: two hundred timings separate two keys whose one-norms differ by five parts in
$370$, and because \MCSI{} uses its static secret in every session, the same key can be
measured again and again. The audit and the measurements in this paper are the specification
of what such an implementation would have to fix, and reproducing them on a replacement is a
matter of running the same makefile targets. Until that piece exists the protocol described
here is a specification with a working reference, and not something anyone should deploy.

\bibliographystyle{unsrt}
\bibliography{mcsi}

\end{document}